\documentclass{lmcs}
\pdfoutput=1
\usepackage[utf8]{inputenc}

\usepackage{lastpage}
\lmcsdoi{21}{4}{9}
\lmcsheading{}{\pageref{LastPage}}{}{}%
{Feb.~16,~2024}{Oct.~16,~2025}{}

\usepackage{tikz}
\usetikzlibrary{shapes.geometric,positioning}
\usepackage{tikz-cd}
\usepackage{cmll}
\usepackage{stmaryrd}
\usepackage{xparse}
\usepackage{mathtools}
\usepackage{amsmath}
\usepackage{amssymb}
\usepackage{amsfonts}
\usepackage{alphabeta}
\usepackage{bm}
\usepackage{ebproof}
\usepackage{wasysym}
\usepackage[shortlabels]{enumitem}
\usepackage[mathscr]{eucal}

\newcommand{\definitive}[1]{\textbf{#1}}
\NewDocumentCommand{\lionel}{O{}m}{\todo[color=orange,#1]{#2 ---~Lionel}}
\NewDocumentCommand{\lison} {O{}m}{\todo[color=yellow,#1]{#2 ---~Lison}}
\NewDocumentCommand{\pierre}{O{}m}{\todo[color=cyan,#1]  {#2 ---~Pierre}}

\newcommand{\eqdef}{\stackrel{\scriptscriptstyle\mathrm{def}}{=}}
\newcommand{\recdef}{\Coloneqq}
\newcommand{\iffdef}{\stackrel{\scriptscriptstyle\mathrm{def}}{\Leftrightarrow}}

\newcommand{\bij}{\simeq}

\DeclarePairedDelimiter{\pars}{(}{)}
\DeclarePairedDelimiter{\mset}{[}{]}
\DeclarePairedDelimiter{\tuple}{\langle}{\rangle}
\DeclarePairedDelimiter{\size}{|}{|}

\newcommand{\N}{\mathbf N}

\newcommand{\WordsOf}[1]{#1^*}

\newcommand{\Mf}{\mathscr{B}}
\newcommand{\MfOf}[2][]{\Mf\pars[#1]{#2}}

\newcommand{\Sf}{\mathscr{S}}
\newcommand{\SfOf}[2][]{\Sf\pars[#1]{#2}}

\newcommand{\fSumsOf}[2][]{\Sigma\pars[#1]{#2}}

\NewDocumentCommand{\card}{sO{}m}{\#\IfBooleanTF{#1}{\pars[#2]{#3}}{#3}}
\newcommand{\supp}{\mathsf{supp}}

\newcommand{\seq}[1]{\vec{#1}}
\newcommand{\emptyseq}{\tuple{}}

\newcommand{\bag}[1]{\bar{#1}}
\newcommand{\emptybag}{\mset{}}
\newcommand{\bagcons}{*}

\newcommand{\splitinto}{\lhd}
\newcommand{\darksplitinto}{\LHD}
\newcommand{\restrict}{\mathbin\upharpoonright}

\DeclarePairedDelimiter{\intr}{\llbracket}{\rrbracket}
\DeclarePairedDelimiter{\sintr}{\|}{\|}
\DeclarePairedDelimiter{\ev}{|}{|}

\newcommand{\pol}{\mathrm{pol}}
\newcommand{\imc}{\rightarrowtriangle}
\newcommand{\iso}{\cong}
\newcommand{\sym}{\cong}

\newcommand{\tto}{\Rightarrow}
\newcommand{\tensor}{\otimes}

\newcommand{\conf}{\mathcal{C}}
\newcommand{\init}{\mathsf{init}}
\newcommand{\display}{\partial}

\newcommand{\vide}{\varnothing} 

\newcommand{\pos}{\mathscr{P}}
\newcommand{\ic}[1]{\overline{#1}}
\newcommand{\rep}[1]{\underline{#1}}

\DeclarePairedDelimiter{\deseq}{\llparenthesis}{\rrparenthesis}
\newcommand{\Aug}{\mathsf{Aug}}
\newcommand{\IAug}{\mathsf{Isog}}

\newcommand{\augOne}{{1}}
\newcommand{\lift}{\square}

\newcommand{\inter}{\circledast}
\newcommand{\causes}{\vartriangleright}
\newcommand{\rside}{\mathtt{r}}
\newcommand{\lside}{\mathtt{l}}
\newcommand{\just}{\mathsf{just}}
\newcommand{\depth}{\mathsf{depth}}

\newcommand{\blambda}{\bm{\lambda}}

\newcommand{\bdelta}{\bm{\delta}}
\newcommand{\bmu}{\bm{\mu}}

\newcommand{\rp}{\overline{\mathbb{R}}_+}
\newcommand{\augCC}{{c\!c}}
\newcommand{\isogCC}{\mathsf{c\!c}}

\newcommand{\id}{\mathrm{id}}
\newcommand{\Sym}{\mathsf{Sym}}
\NewDocumentCommand{\SymOf}{sO{}m}{\Sym\IfBooleanTF{#1}{\pars*{#3}}{\pars[#2]{#3}}}
\NewDocumentCommand{\nSymOf}{sO{}m}{\card{\Sym\IfBooleanTF{#1}{\pars*{#3}}{\pars[#2]{#3}}}}

\newcommand{\Strat}{\mathsf{Strat}}
\newcommand{\C}{\mathcal{C}}
\newcommand{\pid}{\id^{\bullet}}
\newcommand{\evm}{\mathrm{ev}}
\newcommand{\var}{\mathrm{var}}
\newcommand{\pack}[1]{\langle\!| #1 |\!\rangle} 
\newcommand{\distribute}{\mathsf{d}}
\newcommand{\rassemble}{\mathsf{g}}
\newcommand{\bdistribute}{\mathbf{d}}
\newcommand{\brassemble}{\mathbf{g}}

\tikzset{
    rcat/.style={
        scale=.7,
        every node/.style={font=\small,inner sep=1pt},
    }
}

\usetikzlibrary{arrows.meta}
\tikzset{
    move/.style={circle,inner sep=0},
    sdep/.style={-,dotted, bend left=15,looseness=.7}, 
    ddep/.style={<-,>={Triangle[length=1ex,width=1ex]}}, 
}
\tikzcdset{
    gs/.style={
        cramped, cells={move},
        column sep=2pt, row sep=1ex, 
    },
    gs header/.style={gs,/tikz/row 1/.style={rectangle}}, 
}

\NewDocumentCommand{\labs}{m O{.} s O{} m}{%
  \lambda #1#2%
  \IfBooleanT{#3}{#4(}#5\IfBooleanT{#3}{#4)}%
}
\NewDocumentCommand{\appl}{s O{} m s O{} m}{%
  \IfBooleanT{#1}{#2(}#3\IfBooleanT{#1}{#2)}\,%
  \IfBooleanT{#4}{#5(}#6\IfBooleanT{#4}{#5)}%
}
\let\rappl=\appl

\NewDocumentCommand{\rsubst}{s O{} m m O{} m}{%
  \IfBooleanT{#1}{#2(}#3\IfBooleanT{#1}{#2)}#5\langle#6/#4#5\rangle%
}
\newcommand{\ssubst}[3]{#1 \langle\hspace{-3pt}\langle #3 / #2 \rangle\hspace{-3pt}\rangle}

\newcommand{\bred}{\leadsto}

\newcommand\ResTerms{\Delta}
\newcommand\ResBags{\MfOf{\ResTerms}}

\newcommand{\rulename}[1]{(\mathrm{#1})}
\ebproofset{right label template=\small$\rulename{\inserttext}$}

\newcommand{\X}{\mathsf{X}}
\newcommand{\Trm}{\mathsf{Tm}}
\newcommand{\Bag}{\mathsf{Bg}}
\newcommand{\Seq}{\mathsf{Sq}}

\NewDocumentCommand{\jugX}{O{\X}}{\vdash_{#1}}
\newcommand{\jugTrm}{\jugX[\Trm]}
\newcommand{\jugBag}{\jugX[\Bag]}
\newcommand{\jugSeq}{\jugX[\Seq]}

\NewDocumentCommand{\SX}{O{\mathsf{X}}}{\Sigma#1}
\newcommand{\STrm}{\SX[\Trm]}

\NewDocumentCommand{\XOf}  {O{\X}mm}{#1(#2; #3)}
\NewDocumentCommand{\XOfNf}{O{\X}mm}{\XOf[#1_{\mathsf{nf}}]{#2}{#3}}
\newcommand{\TrmOf}[2]    {\XOf  [\Trm]{#1}{#2}}
\newcommand{\TrmOfNf}[2]  {\XOfNf[\Trm]{#1}{#2}}
\newcommand{\BagOf}[2]    {\XOf  [\Bag]{#1}{#2}}
\newcommand{\BagOfNf}[2]  {\XOfNf[\Bag]{#1}{#2}}
\newcommand{\SeqOf}[2]    {\XOf  [\Seq]{#1}{#2}}
\newcommand{\SeqOfNf}[2]  {\XOfNf[\Seq]{#1}{#2}}

\NewDocumentCommand{\SXOf}{O{\X}mm}{\XOf[\SX[#1]]{#2}{#3}}
\newcommand{\STrmOf}[2] {\SXOf[\Trm]{#1}{#2}}

\newcommand\varA{x}
\newcommand\varB{y}
\newcommand\varC{z}

\newcommand\termA{s}
\newcommand\termB{t}
\newcommand\termC{u}

\newcommand\bagA{\bag{\termA}}
\newcommand\bagB{\bag{\termB}}
\newcommand\bagC{\bag{\termC}}

\newcommand\termsA{S}
\newcommand\termsB{T}
\newcommand\termsC{U}

\newcommand\bagsB{\bag{\termsB}}

\newcommand\typeA{A}
\newcommand\typeB{B}
\newcommand\typeC{C}
\newcommand\typeo{o}

\newcommand\vtypeA{\vec\typeA}
\newcommand\vtypeB{\vec\typeB}
\newcommand\vtypeC{\vec\typeC}

\newcommand\vvarA{\vec\varA}

\newcommand\vbagA{\seq{\termA}}
\newcommand\vbagB{\seq{\termB}}
\newcommand\vbagC{\seq{\termC}}

\renewcommand{\qu}{\mathtt{q}}

\newcommand{\posX}{\mathsf{x}}
\newcommand{\posY}{\mathsf{y}}
\newcommand{\posZ}{\mathsf{z}}
\newcommand{\posU}{\mathsf{u}}
\newcommand{\posV}{\mathsf{v}}

\newcommand{\augQ}{{{q}}}
\newcommand{\augP}{{{p}}}
\newcommand{\augR}{{{r}}}

\newcommand{\isogQ}{\mathsf{q}}
\newcommand{\isogP}{\mathsf{p}}
\newcommand{\isogR}{\mathsf{r}}

\newcommand{\strS}{\sigma}
\newcommand{\strT}{\tau}

\newcommand{\starplus}{\mathop{\cup\hspace{-6.4pt}*}}
\newcommand{\smallstarplus}{\mathop{\cup\hspace{-5pt}*}}

\begin{document}
	
\title[Strategies as Resource Terms, and Their Categorical Semantics]{Strategies as Resource Terms,\texorpdfstring{\\}{} and Their Categorical Semantics}
\titlecomment{{\lsuper*}This paper is an extended version of \cite{BCVA23}.}

\author[L.~Blondeau-Patissier]{Lison Blondeau-Patissier}[a,b,c]
\author[P.~Clairambault]{Pierre
	Clairambault\lmcsorcid{https://orcid.org/0000-0002-3285-6028}}[b]
\author[L.~Vaux Auclair]{Lionel Vaux
	Auclair\lmcsorcid{https://orcid.org/0000-0001-9466-418X}}[c]	
	
\address{ENS de Lyon, CNRS, Université Claude Bernard Lyon~1, LIP, Lyon, France}
\email{lison.blondeau-patissier@ens-lyon.org}

\address{Aix Marseille Univ, CNRS, LIS, Marseille, France}
\email{pierre.clairambault@cnrs.fr}
\thanks{
	\emph{Pierre Clairambault}:~Supported by the ANR projects 
	DyVerSe (ANR-19-CE48-0010-01) 
	and by the Labex MiLyon (ANR-10-LABX-0070) 
	of Universit\'e de Lyon.}
	
\address{Aix-Marseille Univ, CNRS, I2M, Marseille, France}
\email{lionel.vaux@univ-amu.fr}
\thanks{
	\emph{Lionel Vaux Auclair}:~Supported by the ANR projects
	PPS (ANR-19-CE48-0014), LambdaComb (ANR-21-CE48-0017)
	and Reciprog (ANR-21-CE48-0019).
	}

\ACMCCS{Theory of computation $\to$ Denotational semantics, 
	    Categorical semantics}

\begin{abstract}
	As shown by Tsukada and Ong, simply-typed, normal and $\eta$-long resource
	terms correspond to plays in Hyland-Ong games, quotiented by Melliès'
	homotopy equivalence.
        The original proof of this inspiring result is indirect, relying on the
        injectivity of the relational model with respect to both sides of the
        correspondence
	-- in particular, the dynamics of the resource calculus is taken into
	account only \emph{via} the compatibility of the relational model with 
	the composition of normal terms defined by normalization.
	
	In the present paper, we revisit and extend these results.
	Our first contribution is to restate the correspondence by
	considering causal structures we call \emph{augmentations},
	which are canonical representatives of Hyland-Ong plays up to homotopy.
	This allows us to give a direct and explicit
	account of the connection with normal resource terms. As a second
	contribution, we extend this account to the \emph{reduction} of
	resource terms: building on a notion of strategies as weighted sums
	of augmentations, we provide a denotational model of the resource
	calculus, invariant under reduction. A key step -- and our third
	contribution -- is a categorical model we call a \emph{resource
		category}, which is to the resource calculus what differential
	categories are to the differential $\lambda$-calculus.
\end{abstract}

\maketitle

\section{Introduction}

The \emph{Taylor expansion} of $λ$-terms translates programs with
possibly infinite behaviour to infinite formal sums of terms of a
language with a strongly finitary behaviour called the \emph{resource
	calculus}. Its discovery dates back to Ehrhard and Regnier's
\emph{differential $\lambda$-calculus}~\cite{ER03},
reifying syntactically the features of particular models of linear logic where
types are interpreted as particular topological vector spaces,
and proofs as linear and continuous maps between them.
Since its inception~\cite{ER08}, Taylor expansion was intended as a
quantitative alternative to order-based approximation techniques, such
as Scott continuity and Böhm trees.  For instance, Barbarossa and
Manzonetto leveraged it to get simpler proofs of known results in
pure $\lambda$-calculus~\cite{BM20}.

\emph{Game semantics} is another well-established line of work, also
representing programs as collections of finite behaviours.  It is
particularly well known for its many full abstraction 
results~\cite{HO00,AJM13}.
How different is the Taylor expansion of the $\lambda$-calculus from its game
semantics? Not very different, suggest Tsukada and Ong~\cite{TO16},
who show that certain normal and
$\eta$-long resource terms correspond bijectively to plays in
the sense of Hyland-Ong game semantics~\cite{HO00}, up to Opponent's scheduling of
the independent explorations of separate branches of the term, as
formalized by Melliès' homotopy equivalence on plays~\cite{M06}.

The account of this insightful result by Tsukada and Ong is inspiring, 
but it also comes with limitations.
Their focus is on normal resource terms, and the dynamics is treated
only in the form of the composition of terms,
\emph{i.e.}\ substitution followed by normalization.
The correspondence is also very indirect, relying on the injectivity of the
relational model {w.r.t.} both normal resource terms and plays up to homotopy.
In~\cite{TO16}, after laying out the intuitions supporting 
the correspondence, Tsukada and Ong motivate this indirect route, writing:
``The idea will now be intuitively clear. However the definition
based on the above argument, which heavily depends on graphical
operations, does not seem so easy to handle.''

In the present paper, we handle this very task.
We rely on a representation of
plays quotienting out Opponent's scheduling,
recently introduced by the first two authors~\cite{BC21}. 
This was inspired by concurrent
games~\cite{CCRW17} --
similar causal structures existed before, first
suggested in~\cite{L05},
and fleshed out more in~\cite{ST17}.
In~\cite{BC21}, plays are replaced by
so-called \emph{augmentations}, which \emph{augment} valid states of
the game with causal constraints imposed by the program. Our first
contribution is an explicit description of the bijection between
normal resource terms and isomorphism classes of augmentations (called
\emph{isogmentations}, for the sake of brevity), in a style
similar to traditional finite definability arguments.

We moreover strive to account for non-normal resource terms and \emph{reduction}
in the resource calculus.
In game semantics, this typically relies on a category of strategies,
whose composition is defined by interaction between plays.
Considering the interaction of augmentations
-- which was not addressed in~\cite{BC21} --
an interesting phenomenon occurs.
Indeed, there is no canonical way to \emph{synchronize} two augmentations:
they can only interact \emph{via} a mediating map, called a \emph{symmetry},
and the result of the interaction depends on the chosen symmetry!
Composition is then obtained by summing over all symmetries.
This is not an artificial phenomenon arising from our implementation
choices: it is analogous to the non-determinism inherent to the
substitution of resource terms. And this is instrumental in our second
contribution: the correspondence between
normal resource terms and isogmentations refines into a denotational
interpretation, invariant under reduction, of resource terms as
``strategies'' -- weighted sums of isogmentations.

To establish this result we expose the structure of the category of strategies
that is relevant to obtain a model of the resource calculus:
we call \emph{resource categories} the resulting categorical model,
which is our third contribution.
We then show that strategies indeed form a resource category, completing the
proof of the previous point.

\paragraph{Related and future work.}
As mentioned above, Tsukada and Ong~\cite{TO16} 
considered some dynamic aspects of the correspondence:
they proved their bijection is compatible with the compositions of terms and
plays, \emph{via} composition in the relational model.
Nonetheless, they did not consider an interpretation of \emph{non-normal}
resource terms as strategies: the question of invariance under reduction
could not even be formulated, and the relevant structure of the category of
strategies could not be exposed.
Still, they state that the normal form of the Taylor
expansion of a $\lambda$-term is isomorphic to its game semantics. 

Our contributions in the present paper constitute a first step to flesh out
this isomorphism into the commuting diagrams:
\[
    \begin{tikzcd}
            M\arrow[|->,"\mathcal T"]{r}\arrow[|->,"\intr{-}"']{d}\arrow[phantom,"(a)"]{rd}&
            \mathcal T(M) \arrow[|->,"\mathcal N"]{r}\arrow[|->,"\intr{-}"]{d}\arrow[phantom,"(b)"]{rd}&
            \mathcal N(\mathcal T(M))\arrow[phantom,"\simeq" sloped]{d}\\
            \intr{M}\arrow[phantom,"="]{r}&\intr{\mathcal T(M)}\arrow[phantom,"="]{r}&\intr{\mathcal N(\mathcal T(M))}
    \end{tikzcd}
    \qquad
    \begin{tikzcd}
            s\arrow[|->,"\mathcal N"]{r}\arrow[|->,"\intr{-}"']{d}\arrow[phantom,"(c)"]{rd}&
            \mathcal N(s)\arrow[phantom,"\simeq" sloped]{d}\\
            \intr{s}\arrow[phantom,"="]{r}&\intr{\mathcal N(s)}
    \end{tikzcd}
\]
where $M$ is a $\lambda$-term, $s$ is a resource term, $\mathcal T$ is Taylor expansion, 
$\mathcal N$ is normalization and $\intr-$ is game semantics. The main
outcome of this paper is to establish the commutation of $(c)$, which appears in our paper under the form of
\autoref{cor:main} -- here, $\simeq$ is the bijection between
weighted sums of normal resource terms and \emph{strategies}, obtained
by lifting to weighted sums our bijection between normal resource terms
and isogmentations. The commutation of $(b)$ follows just using the
linearity of $\intr{-}$ and $\mathcal N$.
In contrast, $(a)$ does not follow from this paper. We expect that it
should be derived from the inductive definitions of both Taylor
expansion and the game semantics interpretation, showing that the
interpretation of each ($η$-long) $λ$-term constructor matches the
Taylor expansion of that constructor; but we prefer to keep this paper
focused on the finitary resource calculus, so we leave that for future
work (see also the next paragraphs and \autoref{sec:conclusion}). 

A significant aspect of our contributions is to take
coefficients into account. This is far
from anecdotal: it requires new methods (we cannot get that \emph{via} the
relational model), it makes the development significantly more complex,
and it is necessary if one expects to apply these tools to a
quantitative setting (e.g., with probabilities) and to provide the basis
of a full game semantical account of quantitative Taylor expansion.

The exact relationship between differential categories~\cite{BCS06,BCLS20} and
resource categories is the subject of ongoing work.
The construction of a resource category from a differential category (more
precisely, from an additive monoidal storage category with
codereliction~\cite{BCLS20}) has been exposed by the first author
~\cite{B24}, although the proposed construction does not preserve
closure.
The reverse direction is less clear.
Indeed, let us stress the fact that the resource calculus is the finitary
fragment of the differential $\lambda$-calculus:
it does not contain the pure $\lambda$-calculus. 
Accordingly, models of the resource calculus are rather related to those of
promotion-free differential linear logic \cite{ER06}:
the exponential modality ($\oc$) need not be a comonad.
From such a model, one can recover an interpretation of the
full differential $\lambda$-calculus \emph{via} Taylor expansion, \emph{provided the
	necessary infinite sums are available}.
So we are convinced (see our concluding remarks in \autoref{sec:conclusion}) that
our category of games does induce a cartesian closed differential
category~\cite{BCS09,BEM10,M12};
more generally, we plan to study how this generalizes to any resource category --
provided the necessary sums of morphisms are available.

\paragraph{Outline of the paper.}
After introducing preliminary notations in \autoref{sec:tuples},
we review the resource calculus in \autoref{sec:syntax};
in particular, we consider a simply typed version whose normal forms are
$η$-long, meaning that each variable is fully applied.
In \autoref{sec:resource_as_aug} we introduce \emph{augmentations},
and show that normal resource terms are in bijective correspondence
with \emph{isogmentations} \emph{i.e.}\ isomorphism classes of augmentations.
We moreover introduce a category of \emph{strategies},
which are weighted sums of isogmentations,
in \autoref{sec:strategies}.
In \autoref{sec:resource_categories} we introduce resource categories, define
the interpretation of the resource calculus in a resource category, and prove
that it is invariant under reduction.
We then show that strategies form a resource category in
\autoref{sec:strategies_resource_cat};
we moreover show in \autoref{subsec:compat_nf} that, up to the bijection between
normal terms and isogmentations, the interpretation of resource terms as
strategies coincides with normalization.
In \autoref{sec:conclusion}, we conclude by discussing the role and side-effects
of $η$-expansion in the resource calculus, in the light of game semantics and
Taylor expansion.

\paragraph{Differences with the conference version.}
A short version of this article~\cite{BCVA23} was published in the
proceedings of the FSCD 2023 conference~\cite{FSCD2023}.
Compared with that former version, the present paper provides a more detailed
account of our constructions, and technical lemmas that were previously left
out are included in appendices.
Moreover, the resource calculus we consider is slightly different:
instead of fully applying each variable and each redex, as was also done by
Tsukada and Ong~\cite{TO16}, we only impose this constraint on variables.
It turns out that this simplifies the exposition, and still yields the same
$η$-long normal forms.

\section{Preliminaries on tuples and bags}
\label{sec:tuples}

If $X$ is a set, we write $\WordsOf{X}$ for the set of finite lists, or tuples,
of elements of $X$, ranged over by $\vec{a},\vec{b}$, \emph{etc.}
We write $\vec{a}=\tuple{a_1,\dotsc,a_n}$ to list the elements of $\vec a$, of
length $\size{\vec a}=n$.
The empty list is $\emptyseq$, and concatenation is simply juxtaposition,
\emph{e.g.}, $\vec{a}\vec{b}$.

We write $\Mf(X)$ for the set of finite multisets of elements of $X$, which we
call \definitive{bags}, ranged over by $\bag{a}, \bag{b}$, \emph{etc.}
We write $\bag{a}=\mset{a_1, \dots, a_n}$ for the bag induced by the list
$\vec{a}=\tuple{a_1,\dotsc,a_n}$ of elements: we then say $\vec{a}$ is an
\definitive{enumeration} of $\bag{a}$ in this case.
We write $\emptybag$ for the empty bag, and use $\bagcons$ for bag
concatenation.
We also write $\size{\bag{a}}$ for the length of $\bag{a}$,
\emph{i.e.}\ the length of any enumeration of $\bag{a}$.

We shall often need to \emph{partition} bags, which requires some care.
For $\bag{a}\in\MfOf{X}$ and $k\in\N$, a \definitive{$k$-partitioning} of $\bag{a}$
is a function $p:\{1,\dotsc,\size{\bag{a}}\}\to\{1,\dotsc,k\}$:
we write $p:\bag{a}\splitinto k$.
Given an enumeration $\tuple{a_1,\dotsc,a_n}$ of $\bag{a}$,
the associated \definitive{$k$-partition} is the tuple
$\tuple{\bag{a}\restrict_p 1,\dotsc,\bag{a}\restrict_p k}$,
where we set $\bag{a}\restrict_p i = [a_j \mid p(j) = i]$ for $1\le i \le k$,
so that $\bag{a}=\bag{a}\restrict_p 1\bagcons\cdots\bagcons\bag{a}\restrict_p k$.
The obtained $k$-partition does depend on the chosen enumeration of $\bag{a}$
but, for any function $f : \Mf(X)^k \to \mathcal{M}$ with values in a
commutative monoid $\mathcal{M}$ (noted additively), the finite sum
\begin{eqnarray}
\sum_{\bag{a} \splitinto \bag{a}_1\bagcons\cdots\bagcons\bag{a}_k}
f(\bag{a}_1,\dotsc,\bag{a}_k)
&\eqdef&
\sum_{p : \bag{a} \splitinto k} f(\bag{a}\restrict_p 1, \dotsc,
\bag{a}\restrict_p k)
\quad\in\quad\mathcal{M}
\label{eq:splitinto}
\end{eqnarray}
is independent from the enumeration.
When indexing a sum with
${\bag{a} \splitinto \bag{a}_1\bagcons\cdots\bagcons\bag{a}_k}$
we thus mean to sum over all partitionings $p : \bag{a} \splitinto k$,
using $\bag{a}_i$ as a shorthand for $\bag{a}\restrict_p i$ in each summand.

We will also consider tuples of bags:
we write $\SfOf{X}$ for $\WordsOf{\MfOf{X}}$.
We denote elements of $\SfOf{X}$ as $\vec{a},\vec{b}$, \emph{etc.} just like
for plain tuples,
but we reserve the name \definitive{sequence} for such tuples of bags.

\section{The Simply-Typed $\eta$-Expanded Resource Calculus}
\label{sec:syntax}

\paragraph{Resource calculus.}
The terms of the resource calculus \cite{ER08},
called \definitive{resource terms}, are just like ordinary $\lambda$-terms,
except that the argument in an application is a bag of terms instead of just
one term.
\begin{figure}
    \begin{align*}
      \ResTerms \ni \termA,\termB,\termC,\ldots
      &
      \recdef x\mid \labs{x}{\termA} \mid \rappl{\termA}{\bagB}
      \\
      \ResBags \ni \bagA,\bagB,\bagC,\ldots
      &
      \recdef \mset{\termA_1,\dotsc,\termA_n}
    \end{align*}
    \caption{Syntax of the resource calculus: resource terms and resource bags}
    \label{fig:resourcecalculus}
\end{figure}
We recall the syntax in \autoref{fig:resourcecalculus}:
we denote resource terms by $\termA,\termB,\termC$,
and bags of terms, which we call \definitive{resource bags},
by $\bagA,\bagB,\bagC$, possibly with sub- and superscripts, 
and write $\ResTerms$ for the set of terms.

The dynamics relies on a multilinear variant of substitution, that we will call
\definitive{resource substitution}:
a redex $\rappl*{\labs{\varA}{\termA}}{\bagB}$ reduces to a formal finite
sum $\rsubst{\termA}{\varA}{\bagB}$ of terms, each summand being obtained by
substituting each element of $\bagB$ for exactly one occurrence of $\varA$ in
$\termA$.
The inductive definition is in \autoref{fig:subst},
\begin{figure}[t]
\begin{align*}
  \rsubst{\varB}{\varA}{\bagB}
  & \eqdef \begin{cases}
    \termB&\text{if $\varB=\varA$ and $\bagB=\mset{\termB}$}\\
    \varB&\text{if $\varB\not=\varA$ and $\bagB=\emptybag$}\\
    0&\text{otherwise}\\
  \end{cases}\\
  \rsubst*{\labs{\varC}{\termA}}{\varA}{\bagB}
  &\eqdef
  \labs{\varC}{\rsubst{\termA}{\varA}{\bagB}}
  \\
  \rsubst*{\rappl{\termA}{\bagC}}{\varA}{\bagB}
  &\eqdef
  \sum_{\bagB\splitinto\bagB_1\bagcons\bagB_2}
  \rappl*{\rsubst{\termA}{\varA}{\bagB_1}}*{\rsubst{\bagC}{\varA}{\bagB_2}}
  \\
  \rsubst{\mset{\termA_1,\dotsc,\termA_n}}{\varA}{\bagB}
  &\eqdef
  \sum_{\bagB\splitinto\bagB_1\bagcons\cdots\bagcons\bagB_n}
  \mset{\rsubst{\termA_1}{\varA}{\bagB_1},\dotsc,\rsubst{\termA_n}{\varA}{\bagB_n}}
\end{align*}
\caption{Inductive definition of substitution ($z$ is chosen fresh in
the abstraction case)}
\label{fig:subst}
\end{figure}
relying on an extension of syntactic constructs to finite sums of expressions:
\[
\labs{\varA}{\termsA} \eqdef\sum_{i\in I}\labs{\varA}{\termA_i}
\qquad
\mset{\termsA}\bagcons\bagsB \eqdef\sum_{i\in I} \sum_{j\in J}
\mset{\termA_i}\bagcons{\bagB_j}
\qquad
\rappl{\termsA}{\bagsB} \eqdef\sum_{i\in I} \sum_{j\in J}
\rappl{\termA_i}{\bagB_j}\,,
\]
for $\termsA = \sum_{i\in I} \termA_i$ and $\bagsB= \sum_{j\in J} \bagB_j$.
Note that the application (resp.\ bag) case in \autoref{fig:subst}
uses the notation introduced in \eqref{eq:splitinto},
the target commutative monoid being the set of formal finite sums of
resource terms (resp.\ of resource bags).
Also note that the definition of \(
    \rsubst{\mset{\termA_1,\dotsc,\termA_n}}{\varA}{\bagB}
\) does not depend on the chosen enumeration
\(
    \tuple{\termA_1,\dotsc,\termA_n}
\).
Resource substitution is in turn extended by linearity, setting
$\rsubst{\termsA}{\varA}{\bagsB}
\eqdef\sum_{i\in I}\sum_{j\in J}\rsubst{\termA_i}{\varA}{\bagB_j}$
with the same notations as above.

The actual protagonists of the calculus are thus sums of terms rather than
single terms.
We will generally write $\fSumsOf{X}$ for the set of finite formal sums on set $X$
--
one may equivalently consider $\fSumsOf{X}$ as the free \(\N\)-module generated
by $X$ and define sums by taking (finite) linear combinations of elements
of $X$, with weights in \(\N\), but any such linear combination
can be written as a plain sum;
those may also be considered as finite multisets, but we adopt a distinct
additive notation to avoid confusion with bags.

\begin{figure}[t]
\[
  \begin{prooftree}[center=false]
    \infer0{\rappl*{\labs{\varA}{\termA}}{\bagB}\bred\rsubst{\termA}{\varA}{\bagB}}
  \end{prooftree}
  \qquad
  \begin{prooftree}[center=false]
    \hypo{\termA\bred\termsA'}
    \infer1{\labs{\varA}{\termA}\bred\labs{\varA}{\termsA'}}
  \end{prooftree}
  \qquad
  \begin{prooftree}[center=false]
    \hypo{\termA\bred\termsA'}
    \infer1{\rappl{\termA}{\bagB}\bred\rappl{\termsA'}{\bagB}}
  \end{prooftree}
  \qquad
  \begin{prooftree}[center=false]
    \hypo{\termA\bred\termsA'}
    \infer1{\mset{\termA}\bagcons{\bagB}\bred\mset{\termsA'}\bagcons{\bagB}}
  \end{prooftree}
  \qquad
  \begin{prooftree}[center=false]
    \hypo{\bagB\bred\bagsB'}
    \infer1{
      \rappl{\termA}{\bagB}
      \bred\rappl{\termA}{\bagsB'}
    }
  \end{prooftree}
\]

\caption{Rules of resource reduction}
\label{fig:reduction}
\end{figure}
The \definitive{reduction of resource terms}
$\mathord{\bred}\subseteq\ResTerms\times\fSumsOf{\ResTerms}$ is
defined inductively by the rules of \autoref{fig:reduction}
-- simultaneously with the reduction of resource bags
$\mathord{\bred}\subseteq\ResBags\times\fSumsOf{\ResBags}$.
It is extended to
$\mathord{\bred}\subseteq\fSumsOf{\ResTerms}\times\fSumsOf{\ResTerms}$
by setting $\termsA\bred\termsA'$ whenever $\termsA=\termB+\termsC$ and
$\termsA'=\termsB'+\termsC$ with $\termB\bred\termsB'$.

\begin{thmC}[\cite{ER08}]
    The reduction $\bred$ on $\fSumsOf{\ResTerms}$ is confluent and
    strongly normalizing.
\end{thmC}

\paragraph{Typing and $\eta$-long terms.}
In the remainder of the paper, we will consider a simply typed version of the
resource calculus, based on the usual grammar of types
\[
  \typeA,\typeB,\ldots \recdef \typeo\mid \typeA \to \typeB
\]
with a single base type $\typeo$.
If $\vtypeA=\tuple{\typeA_1,\dotsc,\typeA_n}$, we write
$\vtypeA\to\typeB\eqdef\typeA_1\to \cdots\to\typeA_n\to\typeB
=\typeA_1\to (\cdots\to(\typeA_n\to\typeB)\cdots)$.
Then any type $\typeC$ can be written uniquely as
$\typeC=\vtypeA\to\typeo$.

The above strong normalization result holds in the untyped setting.
We use typing only to enforce a syntactic constraint on terms,
ensuring that normal resource expressions are in $\eta$-long form:
normal terms of type $\vtypeA\to\typeo$ will necessarily have the shape
$\labs{\varA_1}[\ldots]{\labs{\varA_{\size{\vtypeA}}}{\termA}}$
with $s$ a normal term of type $\typeo$, which must then be
a fully applied variable.
This is required to obtain an exact correspondence between normal terms and
isogmentations, because game semantics is inherently extensional.

We fix a type for each variable, so that each type has infinitely many
variables, and write $\varA:\typeA$ when $\typeA$ is the type of $\varA$.
A typing context $\Gamma$ is a finite set of typed variables. As usual
we write it as any enumeration $\varA_1:\typeA_1,\dotsc,\varA_n:\typeA_n$,
abbreviated as $\vvarA:\vtypeA$; we may then also write
$\labs{\vvarA}{\termA}\eqdef\labs{\varA_1}[\dotsc]{\labs{\varA_n}{s}}$.

We call \definitive{resource sequence} any sequence
$\vbagA\in\SfOf{\ResTerms}=\WordsOf{\MfOf{\ResTerms}}$.
Given a term $\termA$ and a resource sequence
$\vbagB=\tuple{\bagB_1,\dotsc,\bagB_k}$,
we also define the application
$\rappl{\termA}{\vbagB}
\,\eqdef\rappl{\termA}{\bagB_1\cdots\bagB_k}
=\rappl*{\cdots(\rappl{\termA}{\bagB_1})\cdots}{\bagB_k}
$.
We extend resource substitution to sequences by setting
\[
  \rsubst{\tuple{\bagA_1,\dotsc,\bagA_n}}{\varA}{\bagB}
  \eqdef
  \sum_{\bagB\splitinto\bagB_1\bagcons\cdots\bagcons\bagB_n}
  \tuple{\rsubst{\bagA_1}{\varA}{\bagB_1},\dotsc,\rsubst{\bagA_n}{\varA}{\bagB_n}}
\]
so that $
  \rsubst*{\rappl{\termA}{\vbagC}}{\varA}{\bagB}
  =
  \sum_{\bagB\splitinto\bagB_1\bagcons\bagB_2}
  \rappl*{\rsubst{\termA}{\varA}{\bagB_1}}*{\rsubst{\vbagC}{\varA}{\bagB_2}}
$, which generalizes the application case of \autoref{fig:subst}.

\begin{figure}[t]
\begin{gather*}
\begin{prooftree}
  \hypo{\Gamma, \varA:\typeA \jugTrm \termA:\typeB}
  \infer1[abs]{\Gamma \jugTrm \labs{\varA}{\termA}:\typeA \to \typeB}
\end{prooftree}
\qquad
\begin{prooftree}
  \hypo{\Gamma \jugTrm \termA :\typeA\to\typeB}
  \hypo{\Gamma \jugBag \bagB  :\typeA}
  \infer2[app]{\Gamma \jugTrm \rappl{\termA}{\bagB}:\typeB}
\end{prooftree}
\\[1em]
\begin{prooftree}
  \hypo{\Gamma, \varA:\vtypeA\to\typeo \jugSeq \vbagB :\vtypeA}
  \infer1[var]{\Gamma, \varA:\vtypeA\to\typeo \jugTrm \rappl{\varA}{\vbagB}:\typeo}
\end{prooftree}
\\[1em]
\begin{prooftree}
  \hypo{\Gamma\jugTrm\termA_1:\typeA}
  \hypo{\cdots}
  \hypo{\Gamma\jugTrm\termA_n:\typeA}
  \infer3[bag]{\Gamma\jugBag\mset{\termA_1,\dotsc,\termA_n}:\typeA}       
\end{prooftree}
\qquad
\begin{prooftree}
  \hypo{\Gamma\jugBag\bagA_1:\typeA_1}
  \hypo{\cdots}
  \hypo{\Gamma\jugBag\bagA_n:\typeA_n}
  \infer3[seq]{\Gamma\jugSeq\tuple{\bagA_1,\dotsc,\bagA_n}:\tuple{\typeA_1,\dotsc,\typeA_n}}
\end{prooftree}
\end{gather*}
\caption{Typing rules for the simply-typed $η$-long resource calculus}
\label{fig:typing}
\end{figure}

The type system is given in \autoref{fig:typing}:
this is a straightforward adaptation of the usual
simple type system for the ordinary $λ$-calculus,
where bags are typed uniformly (all the elements of a bag 
share the same type), with the additional requirement
that variables are fully applied.
Note that there are three kinds of judgements:
$\Gamma\jugTrm\termA:\typeA$ for terms,
$\Gamma\jugBag\bagA:\typeA$ for bags, and
$\Gamma\jugSeq\vbagA:\vtypeA$ for sequences.
For $\X\in\{\Trm,\Bag,\Seq\}$, we write $\XOf{\Gamma}{\typeA}$
for the set of those expressions $\termA$ \emph{s.t.}
$\Gamma\jugX\termA:\typeA$.
We moreover write $\XOfNf{\Gamma}{\typeA}$ for the elements of
$\XOf{\Gamma}{\typeA}$ that are in normal form.

\begin{lem}\label{lem:typing_normal_forms}
    We have $\termA\in\XOfNf{\Gamma}{\typeA}$ iff
    $\Gamma\jugX\termA:\typeA$ is derivable 
    without using the application rule $\rulename{app}$.
\end{lem}
\begin{proof}
    Given a derivation tree for $\Gamma\jugX\termA:\typeA$ using rule
    $\rulename{app}$ at least once,
    consider a minimal subderivation with this property:
    it must have an instance of $\rulename{app}$ at its root,
    and its premises are derived without $\rulename{app}$.
    The left premise must thus be derived by $\rulename{abs}$:
    we have ruled out $\rulename{app}$, and the conclusion of
    $\rulename{var}$ is never an arrow type.
    We thus obtain a redex.
\end{proof}

\begin{cor}
    If $\termA\in\TrmOfNf{\Gamma}{\typeA\to\typeB}$, then
    we can write $\termA=\labs{\varA}{\termB}$ with
    $\varA:\typeA$ and $\termB\in\TrmOfNf{\Gamma,\varA:\typeA}{\typeB}$.
    And if $\termA\in\TrmOfNf{\Gamma}{\typeo}$, then
    we can write $\termA=\rappl{\varB}{\vbagC}$
    with $\varB:\vtypeC\to\typeo\in\Gamma$
    and $\vbagC\in\SeqOfNf{\Gamma}{\vtypeC}$.
\end{cor}

Now we extend the type system to finite sums of terms,
by setting $\Gamma\jugX\sum_{i\in I}\termA_i:\typeA$
if $\Gamma\jugX\termA_i:\typeA$ for each $i\in I$;
and we write $\SXOf{\Gamma}{\typeA}$ for
$\fSumsOf{\XOf{\Gamma}{\typeA}}$. 
This type system enjoys subject reduction {w.r.t.}~$\bred$.
As is usual, the key result for subject reduction is a substitution lemma.
\begin{lem}
    If $\termA\in\XOf{\Gamma,\varA:\typeB}{\typeA}$ and
    $\bagB\in\BagOf{\Gamma}{\typeB}$ then
    $\rsubst{\termA}{\varA}{\bagB}\in\SXOf{\Gamma}{\typeA}$.
\end{lem}
\begin{proof}
    Straightforward, by mutual induction on the three syntactic kinds.
\end{proof}

\begin{lem}[Subject reduction]
  If $\termsA\in\SXOf{\Gamma}{\typeA}$
  and $\termsA\bred\termsA'$ then
  $\termsA'\in\SXOf{\Gamma}{\typeA}$.
\end{lem}
\begin{proof}
    We first treat the case of 
    $\termsA=\termA\in\XOf{\Gamma}{\typeA}$
    by induction on the definition of the reduction $\termA\bred\termsA'$:
    the case of a redex is by the substitution lemma, and the other cases
    follow by contextuality.
    The extension to sums is straightforward.
\end{proof}

Note that for the rest of this paper, all the resource terms we
consider will be typed in the system above.
By \emph{normal resource term}, we will thus always mean \emph{normal
$\eta$-long typed resource term}.

\section{Resource Terms as Augmentations}
\label{sec:resource_as_aug}

\begin{figure}[t]
    \[
    \mathtt{}
    \begin{tikzcd}[gs header,column sep=0pt,row sep=-3pt]
        ((& o &\tto& o &)\tto(& o &\tto& o &)\tto& o &)\tto& o
        \\[2pt]
        &&&&&&&&&&&\qu^- \\
        &&&&&&&&&\qu^+ \ar[urr,sdep] \\
        &&&\qu^- \ar[urrrrrr,sdep]\\
        &\qu^+   \ar[urr,sdep]\\
        &&&&&&&\qu^- \ar[uuurr,sdep]\\
        &&&&&&&&&\qu^+ \ar[uuuuurr,sdep]\\
        &&&\qu^- \ar[uuuuurrrrrr,sdep]\\
        &\qu^+   \ar[urr,sdep]\\
    \end{tikzcd}
    \qquad
    \begin{tikzcd}[gs header,column sep=0pt,row sep=-3pt]
        ((& o &\tto& o &)\tto(& o &\tto& o &)\tto& o &)\tto& o
        \\[2pt]
        &&&&&&&&&&&\qu^- \\
        &&&&&&&&&\qu^+ \ar[urr,sdep] \\
        &&&\qu^- \ar[urrrrrr,sdep]\\
        &\qu^+   \ar[urr,sdep]\\
        &&&\qu^- \ar[uuurrrrrr,sdep]\\
        &\qu^+   \ar[urr,sdep]\\
        &&&&&&&\qu^- \ar[uuuuurr,sdep]\\
        &&&&&&&&&\qu^+ \ar[uuuuuuurr,sdep]\\
    \end{tikzcd}
    \]
    \caption{Two homotopic plays in HO games}
    \label{fig:ho_plays_homotopy}
\end{figure}

Game semantics represents programs by recording all their
interactions with an execution environment. This is formalized in a
game-theoretic language: a program is represented as a strategy in a
two-player game. The first player (corresponding to the program under
scrutiny) is called \emph{Player}, associated with the \emph{positive}
polarity; while the second (corresponding to the execution environment)
is called Opponent, associated with the \emph{negative} polarity.

\paragraph{Plays in game semantics.}
In Hyland-Ong game semantics \cite{HO00}, this leads to a formalization of
executions as \emph{plays}, drawn in diagrams such as in
\autoref{fig:ho_plays_homotopy}, read temporally
from top to bottom.
Nodes are called \emph{moves} or \emph{events}, negative (from Opponent)
or positive (from Player)
-- each corresponds to a resource call,
and the dotted lines, called \emph{justification pointers}, carry
the hierarchical relationship between those calls.

Both diagrams in \autoref{fig:ho_plays_homotopy} represent plays for
a simply typed $λ$-term of the shape
\[
    \vdash \labs{f^{(o \to o) \to (o \to o) \to
o}}{\appl{f}{(\labs{x^o}{x})\,(\labs{y^o}{\appl{f}{M_1\,M_2}})}}
    : ((o \to o) \to (o \to o) \to o) \to o
\]
with $M_1,M_2$ unspecified terms of type $o\to o$
(in which \(y\) and \(f\) may occur free).
For instance, the first diagram of \autoref{fig:ho_plays_homotopy} reads as
follows:
Opponent starts computation with the initial $\qu^-$, to which Player reacts
with the first $\qu^+$, corresponding to calling $f$.
With $\qu^-$ on the third line, Opponent prompts $f$ to inspect its first
argument, to which Player responds with $\qu^+$ on the fourth line,
corresponding to a call to $x$.
Subsequently, Opponent evaluates the second argument of $f$, which also
has a head occurrence of $f$: Player thus answers by calling
$f$ again but, this time, Opponent does not inspect the arguments
$M_1$ and $M_2$ fed to $f$.
Instead, Opponent evaluates the first argument of the first call to $f$
again. This addresses the same subterm as the third move, and so Player
reacts as before (the fact that Player \emph{must} react in the same
way, if the program under scrutiny is pure, corresponds to the notion
of \emph{innocence} in Hyland-Ong games). The second diagram of
\autoref{fig:ho_plays_homotopy} is very similar, the only difference being
that Opponent starts with evaluating twice the first argument of $f$,
and then evaluates the second one. 

\paragraph{Plays and resource terms.}
In \cite{TO16}, seeking a syntactic counterpart to the \emph{plays} of HO
games, Tsukada and Ong state:
\emph{``plays in HO/N-games are terms of a well-known and important calculus,
the resource calculus''}.
This is natural as both game semantics and the resource calculus are
quantitative and represent resource usage explicitly.
In the first play of \autoref{fig:ho_plays_homotopy}, the first argument of 
the first call to $f$ is evaluated \emph{twice} while the second one is
evaluated \emph{once}; and none of the arguments to the second call to $f$ is
evaluated:
this play is accordingly identified with the resource term
\begin{equation}
    \jugTrm
\labs{f}{\appl{f}{\mset{\labs{x}{x},\labs{x}{x}}\,\mset{\labs{y}{\appl{f}{[]\,[]}}}}}
    : ((o \to o) \to (o \to o) \to o) \to o
    \,.
    \label{eq:res_term}
\end{equation}

However, the other play of \autoref{fig:ho_plays_homotopy} \emph{also}
corresponds to that resource term!

Tsukada and Ong actually establish a bijection between 
(simply typed, normal and $η$-long) resource terms
and plays \emph{up to} Melliès' \emph{homotopy relation}
\cite{M06}, relating plays which,
like those of \autoref{fig:ho_plays_homotopy}, only differ via Opponent's
scheduling.

\begin{figure}[t]
    \[\begin{tikzcd}[gs header]
        ((&&o&&\tto&&o&&)\tto(&o&\tto&o&)\tto&o&)\tto&o
        \\[2pt]
        &&&&&&&&&&&&&&&\qu^-
        \\
        &&&&&&&&&&&&&\qu^+ \ar[urr,sdep] \ar[urr,ddep]
        \\[.6cm]
        &&&&&\qu^- \ar[urrrrrrrr,sdep] \ar[urrrrrrrr,ddep]
        &&\qu^- \ar[urrrrrr,sdep] \ar[urrrrrr,ddep]
        &&&&\qu^- \ar[urr,sdep] \ar[urr,ddep]
        \\[.4cm]
        &\qu^+   \ar[urrrr,sdep] \ar[urrrr,ddep]
        &&\qu^+   \ar[urrrr,sdep] \ar[urrrr,ddep]
        &&&&&&&&&&\qu^+ \ar[uuurr,sdep] \ar[ull,ddep]
        \\
    \end{tikzcd}\]
    \caption{An augmentation}
    \label{fig:augmentations}
\end{figure}

The quotient of plays up to homotopy equivalence can be given
an explicit representation.
In \cite{BC21}, Blondeau-Patissier and
Clairambault introduced a \emph{causal} representation of innocent
strategies
(inspired from \emph{concurrent games}
\cite{CCRW17, CCW19}
-- see also \cite{ST17})
as a technical tool to prove a positional injectivity
theorem for innocent strategies.
There, a strategy is not a set of plays, but instead gathers diagrams such as
in \autoref{fig:augmentations}, in which the trained eye can read
exactly the same data as in the resource term \eqref{eq:res_term}:
the model replaces the chronological \emph{plays} of game semantics with
causal structures called \emph{augmentations}, of which the plays are
just particular linearizations. Thus as the first contribution of this
paper, we recast Tsukada and Ong's correspondence as a bijection between
normal resource terms and augmentations.

\subsection{Arenas, Configurations, Augmentations}
\label{subsec:def_aug}
Before defining \emph{augmentations} (and their isomorphism classes,
\emph{isogmentations}), we need to introduce \emph{arenas} and
\emph{configurations} (and their isomorphism classes, \emph{positions}).
As with HO games, arenas correspond to types: 
they represent sets of available computational events for Opponent and Player 
given the type, and they feature a static hierarchy between events, playing a
similar role to the enabling relation of HO arenas.
Positions, \emph{i.e.}\ configurations \emph{up to symmetry},
are the collapse of plays in the relational model:
configurations feature only events along with a \emph{static} causal
order representing the hierarchical dependencies between function
calls, without any temporal information.  Finally, augmentations are
configurations ``augmented'' with another, \emph{dynamic} causal order,
corresponding to the causal links appearing in normal resource terms --
akin to finite Böhm-like trees~\cite[Section 10.1]{B85}.

\paragraph{Arenas.}
First we introduce the representation of types: \emph{arenas}.
Given a poset \(A=\tuple{\ev{A},\leq_A}\),
we write \(\min(A)\) (resp.\ \(\max(A)\)) for the set of its 
minimal (resp.\ maximal) elements; and
we write $a_1 \imc_A a_2$ when $a_1 <_A a_2$ with no element strictly in
between ($a_1$ is \definitive{immediately} below $a_2$).
We call \definitive{forest} any poset \(A\) such that
the principal ideal \(\{a' \in \ev{A} \mid a' \leq_A a\}\)
of each \(a\in\ev{A}\) is finite and linearly ordered by $\leq_A$.

\begin{defi}\label{def:arena}
An \definitive{arena} is a triple $A = \tuple{\ev{A}, \leq_A, \pol_A}$ where
$\tuple{\ev{A}, \leq_A}$ is a (countable) forest, and $\pol_A :
\ev{A} \to \{-, +\}$ is a \definitive{polarity function},
which is
\begin{center}\begin{tabular}{rl}
    \emph{alternating:}
    & for all $a_1 \imc_A a_2$, $\pol_A(a_1) \neq \pol_A(a_2)$.
\end{tabular}\end{center}
In a \definitive{negative arena}, we moreover require $\pol_A(a) = -$ for all
$a \in \min(A)$.
\end{defi}

Elements of $\ev{A}$ are called \definitive{events} or \definitive{moves} interchangeably.
An \definitive{isomorphism} $\varphi : A \iso B$ of arenas is a bijection
between events preserving and reflecting all the structure.

\begin{figure}[t]
    \[\begin{tikzcd}[gs header,row sep=.5ex,]
        ((& o &\tto& o &)\tto(& o &\tto& o &)\tto& o &)\tto& o
        \\
          &&&&&&&&&&&\qu^-
        \\
        &&&&&&&&&\qu^+ \ar[urr,sdep]
        \\
        &&&\qu^- \ar[urrrrrr,sdep]
        &&&&\qu^- \ar[urr,sdep]
        \\
        &\qu^+   \ar[urr,sdep]
        &&&&\qu^+ \ar[urr,sdep]
        \\
    \end{tikzcd}\]
    \caption{The arena $\intr{(( o \to o )\to( o \to o )\to o )\to o}$}
    \label{fig:ex_arena}
\end{figure}

Arenas present computational events with their causal dependencies:
positive moves for Player, and negative moves for Opponent.
We often annotate moves with their polarity.
When depicting arenas we draw the immediate causality $\imc$ as dotted lines,
read from top to bottom.
The (negative) arena corresponding to the type
$(( o \to o )\to( o \to o )\to o)\to o$ is depicted in \autoref{fig:ex_arena}:
we present the construction of the arena $\intr{\typeA}$ associated 
with a type $\typeA$ in the next paragraph.

\paragraph{Constructions.}
To interpret elaborate datatypes, we define some constructors on
arenas. 
We write $X+Y$ for the disjoint union $(\{1\} \times X) \cup (\{2\}
\times Y)$ of sets.

\begin{defi}\label{def:arena_tensor}
    The \definitive{tensor} of arenas $A_1$ and $A_2$ is defined as follows:
    \[
        \begin{array}{rcl}
                \ev{A_1 \tensor A_2} &~~\eqdef~~& \ev{A_1} + \ev{A_2}\\
                (i, a) \leq_{A_1 \tensor A_2} (j, b) &\iffdef&
                i = j ~~ \wedge ~~ a \leq_{A_i} b\\
                \pol_{A_1 \tensor A_2}(i, a) &\eqdef& \pol_{A_i}(a)\,.
        \end{array}
    \]
\end{defi}

An arena $A$ is \textbf{pointed} if it has exactly one minimal event,
written $\init(A)$. The tensor directly extends to countable arity, and each 
arena decomposes as $A \iso \tensor_{i \in \min(A)} A_i$ with each
$A_i$ pointed with $\init(A_i)=i$.
We set $A^{\bot}$ as $A$ with polarities reversed, and write $A\vdash
B$ for $A^{\bot} \tensor B$.
The arena $A \vdash B$, called the \textbf{hom-game}, is the
ambient game for strategies \emph{from $A$ to $B$}. However, it is not
negative and hence will fail to be an object of our category. To
reinstate negativity we must ensure an additional dependency, as done
in the arrow construction below:

\begin{defi}
    If $A_1$ and $A_2$ are negative arenas with $A_2$ pointed, then the
    \definitive{arrow} $A_1 \tto A_2$ is the pointed negative arena defined as follows:
    \[
        \begin{array}{rcl}
                \ev{A_1\tto A_2} &~~\eqdef~~& \ev{A_1} + \ev{A_2}\\
                (i, a) \leq_{A_1 \tto A_2} (j, b) &\iffdef&
                    (i = j \wedge a \leq_{A_i} b) \vee (i,a)=(2,\init(A_2))\\
                \pol_{A_1 \tto A_2}(1, a) &\eqdef& -\pol_{A_1}(a)
                \\
                \pol_{A_1 \tto A_2}(2, a) &\eqdef& \pol_{A_2}(a)\,.
        \end{array}
    \]

We shall also use the generalization of this construction to $A \tto B$
for $B$ not necessarily pointed. In that case, writing $B \iso
\otimes_{i\in \min(B)} B_i$, we set $A \tto B = \otimes_{i\in \min(B)} A \tto B_i$.
\end{defi}

Finally, $1$ is the empty arena and $o$ has exactly one (negative) move
$\qu$. We interpret types as arenas via $\intr{o} = o$ and $\intr{A\to
B} = \intr{A} \tto \intr{B}$, and contexts via $\intr{\Gamma} =
\otimes_{(x : A) \in \Gamma} \intr{A}$. 

As announced, \autoref{fig:ex_arena} shows the interpretation of
$(( o \to o )\to( o \to o )\to o )\to o$.
Note that several nodes in the figure are drawn with the same label.
We could distinguish them via indices to keep them apart.
But as each move corresponds to an occurrence of an atom in the type, we find
it more convenient to follow the long-standing game semantics convention and
keep moves distinct by attempting to always place them under the corresponding
atom.

\paragraph{Configurations.}
Next we define the \emph{states} reached when playing on arena $A$.
Intuitively, a state is a sub-tree of $A$ but where each branch may be explored
multiple times --
such structures were first introduced by Boudes \cite{B09}
under the name \emph{thick subtrees}.
Here, by analogy with concurrent games \cite{CCRW17}, we call them
\emph{configurations}.

\begin{defi}\label{def:conf}
    A \definitive{configuration} $x \in \conf(A)$ of arena $A$ is $x =
    \tuple{\ev{x}, \leq_x, \display_x}$ such that $\tuple{\ev{x}, \leq_x}$ is a
    finite forest, and where $\display_x$ is a
function $\display_x : \ev{x} \to \ev{A}$ called the
\definitive{display map}, subject to the conditions:
    \[
    \begin{array}{rl}
            \text{\emph{minimality-respecting:}} & \text{for any $a \in \ev{x}$, $a\in\min(x)$
                    iff $\display_x(a)\in\min(A)$,}\\
            \text{\emph{causality-preserving:}} & \text{for all $a_1, a_2 \in
                    \ev{x}$, if $a_1 \imc_x a_2$ then $\display_x(a_1) \imc_A \display_x(a_2)$,}
    \end{array}
    \]
    and $x$ is \definitive{pointed} (and we write $x\in\conf_\bullet(A)$) if it
    has exactly one minimal event.
\end{defi}

A polarity on $x$ is deduced by $\pol(a) =
\pol_A(\display_x(a))$.  We write $a^-$ (resp.\ $a^+$) for $a$ s.t.
$\pol(a) = -$ (resp.\ $\pol(a)=+$).
Ignoring the arrows with triangle heads and keeping only the dotted
lines oriented from top to bottom, the diagram of \autoref{fig:augmentations}
defines a configuration on the arena of \autoref{fig:ex_arena} -- notice that
the branch on the left hand side is explored twice, and so is the first
$\qu^+$.

For $x, y \in \conf(A)$, the sets $\ev{x}$ and $\ev{y}$ are arbitrary
and only related to $A$ via $\display_x$ and $\display_y$ -- their
specific identity is irrelevant. So configurations should be compared
up to \emph{symmetry}: a \definitive{symmetry} $\varphi : x \sym_A y$ is an
order-isomorphism s.t. $\display_y \circ \varphi = \display_x$. Symmetry
classes of configurations are called \definitive{positions}: the set of
positions on $A$ is written $\pos(A)$, and they are ranged over by $\posX,
\posY$, \emph{etc.} (note the change of font).  A position $\posX$ is
\definitive{pointed}, written $\posX \in \pos_\bullet(A)$, if any of its
representatives is. If $x \in \conf(A)$, we write $\ic{x} \in \pos(A)$
for the corresponding position. Reciprocally, if $\posX \in \pos(A)$, we
fix a representative $\rep{\posX} \in \conf(A)$.
In \cite{BC21}, positions were shown to correspond to points in the relational
model (if $o$ is interpreted as a singleton).

If $x \in \conf(A)$ and $y \in \conf(B)$, then $x \tensor y \in \conf(A
\tensor B)$ has events the disjoint union $\ev{x} + \ev{y}$, and
display map inherited.
We define $x \vdash y \in \conf(A \vdash B)$ similarly.
If moreover $A=B$, we define $x*y\in\conf(A)$ with, again,
$\ev{x*y}=\ev{x} + \ev{y}$, and with display map the co-pairing
$[\display_x, \display_y] : \ev{x*y} \to \ev{A}$.

\paragraph{Augmentations.}
We finally define our representation of plays up to homotopy:

\begin{defi}\label{def:augmentation}
An \definitive{augmentation} on arena $A$ is a quadruple $\augQ = \tuple{\ev{\augQ},
        \leq_{\deseq{\augQ}}, \leq_\augQ, \display_\augQ}$, where $\deseq{\augQ} =
\tuple{\ev{\augQ}, \leq_{\deseq{\augQ}},\display_\augQ} \in \conf(A)$, and
$\tuple{\ev{\augQ}, \leq_\augQ}$ is a forest satisfying:
\[
\begin{array}{rcl}
\text{\emph{rule-abiding:}} && \text{for all $a_1, a_2 \in \ev{\augQ}$, if
        $a_1 \leq_{\deseq \augQ} a_2$, then $a_1 \leq_\augQ a_2$,}\\ 
\text{\emph{courteous:}} && \text{for all $a_1 \imc_\augQ a_2$, if
        $\pol(a_1) = +$ or $\pol(a_2) = -$, then $a_1 \imc_{\deseq \augQ}
a_2$,}\\
\text{\emph{deterministic:}} && \text{for all $a^- \imc_\augQ a^+_1$ and
        $a^- \imc_\augQ a^+_2$, we have $a_1 = a_2$,}\\
\text{\emph{$+$-covered:}}
&& \text{for all $a \in\max(\augQ)$, we have
$\pol(a)=+$,}\\
\text{\emph{negative:}}
&\quad& \text{for all $a \in \min(\augQ)$, we have $\pol(a) = -$.}
\end{array}
\]

We then write $\augQ \in \Aug(A)$, and call
$\deseq{\augQ} \in \conf(A)$ the \definitive{desequentialization} of $\augQ$.
Finally, $\augQ$ is \definitive{pointed} if $\min(\augQ)$ is an singleton,
and we then write $\augQ \in \Aug_\bullet(A)$.
\end{defi}

Figure~\ref{fig:augmentations} represents an augmentation $\augQ$ by
showing both relations $\imc_{\deseq \augQ}$ (as dotted lines) and
$\imc_\augQ$ (as arrows with black triangle heads).
An augmentation $\augQ \in \Aug(A)$ \emph{augments} a configuration $x \in
\conf(A)$ by specifying causal constraints imposed by the term: for
each event, the augmentation gives the necessary conditions before it
can be played. Augmentations are analogous to plays: 
\emph{plays} in the Hyland-Ong sense can be recovered via the
alternating linearizations of augmentations
\cite{BC21}.
However, augmentations factor out Opponent's scheduling,
hence give an explicit representation of plays up to homotopy.

\paragraph{Isogmentations.}
Just like for configurations, augmentations should be considered up to isomorphism.
An \definitive{isomorphism} of augmentations $\varphi : \augQ \iso \augP$ is a bijection
preserving and reflecting all the structure.
An \definitive{isogmentation} is an isomorphism class of augmentations,
ranged over by $\isogQ, \isogP$, \emph{etc.}:
we write $\IAug(A)$ (resp.\ $\IAug_\bullet(A)$) for
isogmentations (resp.\ \emph{pointed} isogmentations).
If $\augQ \in \Aug(A)$, we write $\ic{\augQ} \in \IAug(A)$
for its isomorphism class; reciprocally, if $\isogQ \in \IAug(A)$, we
fix a representative $\rep{\isogQ} \in \isogQ$.

\subsection{Isogmentations are Normal Resource Terms}
\label{subsec:isog_resourceterms}

Now we spell out the link between isogmentations and normal resource
terms.
We first show how the structure of each syntactic kind is reflected by
augmentations of the appropriate type: in particular,
terms will be mapped to pointed isogmentations, 
and bags to general isogmentations.
The main result (\autoref{th:bij_aug_term}) follows directly.

\paragraph{Tensors and sequences.} 
To reflect the syntactic formation rule for sequences,
we first show that isogmentations on $A_1 \tensor \dots \tensor A_n$ \emph{are}
tuples of isogmentations on the $A_i$'s.
Consider negative arenas $\Gamma, A_1, \dots, A_n$,
and $\augQ_i \in \Aug(\Gamma \vdash A_i)$ for $1\le i\le n$.
We set $\vec{\augQ}=\tuple{\augQ_i \mid 1 \leq i \leq n} \in
\Aug(\Gamma \vdash \tensor_{1\leq i \leq n} A_i)$ with
\[
\ev{\vec{\augQ}\,} \eqdef \sum_{i=1}^n \ev{\augQ_i}\,,
\qquad
\left\{
\begin{array}{rclcl}
\display_{\vec{\augQ}}(i, m) &\eqdef& (1, g) 
&\quad& \text{if $\display_{\augQ_i}(m) = (1, g)$,}\\
\display_{\vec{\augQ}}(i, m) &\eqdef& (2, (i, a))
&& \text{if $\display_{\augQ_i}(m) = (2, a)$,}
\end{array}
\right.
\]
where $m\in\ev{\augQ_i}$, $g\in\ev{\Gamma}$, $a\in\ev{A_i}$, and
with the two orders $\leq_{\vec{\augQ}}$ and
$\leq_{\deseq{\vec{\augQ}}}$ inherited. It is immediate that
this construction preserves isomorphisms, so that it extends to
isogmentations.

\begin{prop}\label{prop:def_seq}
The previous construction on augmentations induces a bijection
\[
\tuple{-, \dots, -} : \prod_{i=1}^n \IAug(\Gamma \vdash A_i)
\bij
\IAug(\Gamma \vdash \tensor_{1 \leq i \leq n} A_i)
\,.
\]
\end{prop}
\begin{proof}
\emph{Injective.} As an isomorphism must preserve $\imc$
and display maps, any isomorphism
\[
\varphi : \tuple{\augQ_i \mid i \in I} \iso
\tuple{\augP_i \mid i \in I}
\]
decomposes uniquely into a sequence of
$\varphi_i : \augQ_i \iso \augP_i$, as required.

\emph{Surjective.}
Consider $\augQ \in \Aug(\Gamma \vdash \tensor_{1\leq i \leq n} A_i)$.
Since $\augQ$ is a forest, $m \in \ev{\augQ}$ has a unique minimal antecedent,
sent by the display map (via condition \emph{negative}) 
to one of the $A_i$'s -- we
say that $m$ is \emph{above $A_i$}. Defining accordingly $\augQ_i$ as
$\augQ$
restricted to the events above $A_i$, we easily construct an isomorphism
$\augQ \iso \tuple{\augQ_i\mid 1 \leq i \leq n}$ as required.
\end{proof}

\paragraph{Bags and pointedness}
The next step is to reflect the typing rule for bags, 
by showing that isogmentations are \emph{bags of pointed isogmentations}.

We start by describing the corresponding construction.
Consider negative arenas $\Gamma$ and $A$, and $\augQ_1, \augQ_2 \in \Aug(\Gamma
\vdash A)$. We set $\augQ_1 * \augQ_2 \in \Aug(\Gamma \vdash A)$ with events
$\ev{\augQ_1 * \augQ_2} = \ev{\augQ_1} + \ev{\augQ_2}$, 
and display
$\display_{\augQ_1 * \augQ_2}(i, m) = \display_{\augQ_i}(m)$,
and the two orders $\leq_{\augQ_1 * \augQ_2}$ and $\leq_{\deseq{\augQ_1 *
\augQ_2}}$ inherited. This generalizes to an $n$-ary operation in the
obvious way, which preserves isomorphisms. The operation induced on
isogmentations is associative and commutative,
and it admits as neutral element the empty isogmentation
$\augOne \in \IAug(\Gamma \vdash A)$ with (a unique
representative with) no event.

\begin{prop}\label{prop:def_bag}
The previous construction on augmentations induces a bijection
\[
    \mset{ - , \dots ,  - }: \Mf(\IAug_\bullet(\Gamma \vdash A))
\bij
\IAug(\Gamma \vdash A)
\,.
\]
\end{prop}
\begin{proof}
\emph{Injective.}
Pick $\augQ_1, \dots, \augQ_n, \augP_1, \dots, \augP_m\in \Aug_\bullet(\Gamma \vdash A)$.
Because the $\augQ_i$'s and $\augP_i$'s are pointed and
isomorphisms preserve the forest structure, an isomorphism
$\varphi : \augQ_1 * \dots * \augQ_n \iso \augP_1 * \dots * \augP_m$ forces $m = n$
and induces a permutation $\pi$ on $n$ and a family of isos $(\varphi_i
: \augQ_i \bij \augP_{\pi(i)})_{1\leq i \leq n}$, which implies $[\ic{\augQ_i}
\mid 1\leq i \leq n] = [\ic{\augP_i} \mid 1 \leq i \leq n]$ as bags.

\emph{Surjective.} As any $\augQ \in \Aug(\Gamma \vdash A)$ is finite, it
has a finite set $I$ of initial moves.
Since $\augQ$ is a forest, any
$m \in \ev{\augQ}$ is above exactly one initial move. For $i \in I$, write
$\augQ_i \in \Aug_\bullet(\Gamma \vdash A)$ the restriction of $\augQ$ above
$i$; then $\augQ \iso \augQ_1 * \dots * \augQ_n$ as required.
\end{proof}

\paragraph{Currying.}
For negative arenas $\Gamma$, $A$ and $B$, we have a bijection
\[
\Lambda_{\Gamma, A, B} : \Aug(\Gamma \tensor A \vdash B) \bij \Aug(\Gamma \vdash A \tto B)
\]
leaving the augmentation unchanged except for the display map,
which is reassigned following
\[
\begin{array}{rclcl}
\display_{\Lambda(\augQ)}(m) &=& (1, \gamma) &\qquad&
\text{if $\display_{\augQ}(m) = (1, (1, \gamma))$}\\
\display_{\Lambda(\augQ)}(m) &=& (2, (i, (2, b))) &&
\text{if $\display_\augQ(m) = (2, b)$ and $\display_\augQ(\init(m)) =
(2, i)$}\\
\display_{\Lambda(\augQ)}(m) &=& (2, (i,(1, a))) &&
\text{if $\display_\augQ(m) = (1, (2, a))$ and
$\display_\augQ(\init(m)) = (2, i)$}
\end{array}
\]
where $\init(m)$ denotes the unique $m' \in \min(\augQ)$ such that $m'
\leq_\augQ m$ -- note that \(\display_q(m')\) must be of the shape
\((2,i)\) with \(i\in\min(B)\) because \(q\) is negative and rule-abiding.

Since isomorphisms of augmentations are order-isomorphisms and preserve
display maps,  the definition of this bijection is obviously compatible
with isomorphisms, and we obtain:

\begin{prop}\label{prop:def_val}
The previous bijection defined on augmentations induces a bijection
\[\Lambda_{\Gamma, A, B} : \IAug_\bullet(\Gamma \tensor A \vdash B) \bij \IAug_\bullet(\Gamma \vdash A \tto B).\]
\end{prop}

\paragraph{Head occurrence.}
The previous constructions handle the rules $\rulename{seq}$, $\rulename{bag}$
and $\rulename{abs}$ of the type system of \autoref{fig:typing}.
By \autoref{lem:typing_normal_forms}, the only remaining case to treat all typed
normal forms is $\rulename{var}$.

As above, we start with the corresponding construction on augmentations.
We write $\vec{B}\tto o \eqdef B_{1} \tto \dots \tto B_{p} \tto o$
for $\vec{B} = \tuple{B_1, \dots, B_n}$ a tuple of objects,
and $\vec{B}^\otimes \eqdef B_1 \otimes \dots \otimes B_n$.
Consider $\Gamma = A_1 \tensor \dots \tensor A_n$ where each $A_i$ is a
negative arena of the shape $A_i = \vec{B_i} \tto o \iso \vec{B}_i^\otimes \tto o$.
Given $\augQ \in \Aug(\Gamma \vdash \vec{B}_{i}^\otimes)$,
the \definitive{$i$-lifting of $\augQ$},
written $\lift_i(\augQ) \in \Aug_\bullet(\Gamma \vdash o)$,
is the augmentation that after the initial Opponent move, starts by playing
the initial move $\qu_i$ of $A_i$ (which is negative in $A_i$ hence positive in
$\Gamma\vdash o$), then proceeds as $\augQ$.
Writing $\qu_{i,1},\dots,\qu_{i,l}$ for the initial moves of $\augQ$
(which must be mapped to initial moves of $\Gamma \vdash \vec{B}_{i}^\otimes$,
hence of $\vec{B}_{i}^\otimes$ because $\augQ$ is \emph{negative} by definition),
$\lift_i(\augQ)$ may be depicted as in \autoref{fig:lifting}.

\begin{figure}[t]
\[
    \begin{tikzpicture}
        \node[
            draw,fill=black!5!white,
            shape=trapezium,trapezium stretches body,
            minimum width=4cm,text height=1cm] (q) {$\augQ$};
        \node[move,below=0pt of q.top left corner,xshift=1ex] (qi1) {$q_{i,1}^-$};
        \node[move,below=0pt of q.top right corner,xshift=-1ex] (qil) {$q_{i,l}^-$};
        \node[below=0pt of q.north] (qij) {$\cdots$};
        \node[move,above=2ex of qij] (qi) {$q_i^+$};
        \node[move,above=2ex of qi] (q0) {$q^-$};
        \draw (qi)  edge[ddep] (q0)
              (qi1) edge[sdep] (qi)
              (qi1) edge[ddep] (qi)
              (qil) edge[sdep,bend right] (qi)
              (qil) edge[ddep] (qi)
        ;
    \end{tikzpicture}
\]
\caption{Illustration of $\lift_i(\augQ)$}
\label{fig:lifting}
\end{figure}

More formally:
\begin{defi}
The \definitive{$i$-lifting of $\augQ$}, written $\lift_i(\augQ) \in
\Aug_\bullet(\Gamma \vdash o)$, has partial order $\leq_\augQ$ prefixed with
two additional moves, namely $\ominus$ and $\oplus$ with $\ominus \imc \oplus$
and $\oplus \imc m$ for each $m \in \init(\augQ)$.
Its desequentialization is the least partial order containing dependencies
\[
\begin{array}{rclcl}
m &\leq_{\deseq{\lift_i(\augQ)}}& n &\qquad& 
\text{for $m, n \in \ev{\augQ}$ with $m \leq_{\deseq{\augQ}} n$,}\\
\oplus &\leq_{\deseq{\lift_i(\augQ)}}& m&&
\text{for all $m \in \ev{\augQ}$ with $\partial_{\augQ}(m) = (2, a)$ for some
$a\in\ev{\vec{B}_i^\otimes}$,}
\end{array}
\]
and with display map given by the following clauses:\footnote{
  Recall that we write $\qu$ for the only move of the arena $o$.
  Also, we keep the isomorphism $A_i \iso \vec{B}_i^\otimes \tto o$ implicit
  whenever we write moves in $A_i$, so that we can consider
  $(i,(2,\qu))\in\ev{\Gamma}$, and $(i,(1,a))\in\ev{\Gamma}$ whenever
  $a\in\ev{\vec{B}_i^\otimes}$.
}
\[
\begin{array}{rclcl}
\partial_{\lift_i(\augQ)}(\ominus) &\eqdef& (2, \qu)\\
\partial_{\lift_i(\augQ)}(\oplus) &\eqdef& (1, (i, (2, \qu)))\\
\partial_{\lift_i(\augQ)}(m) &\eqdef& (1, a)
&\qquad&
\text{if $\partial_{\augQ}(m) = (1, a)$,}\\
\partial_{\lift_i(\augQ)}(m) &\eqdef& (1, (i, (1, a)))
&&\text{if $\partial_\augQ(m) = (2, a)$,}
\end{array}
\]
altogether defining $\lift_i(\augQ) \in \Aug_\bullet(\Gamma \vdash o)$ as
required.
\end{defi}

This construction again preserves isomorphisms and thus extends to
isogmentations:
for any $\isogQ \in \IAug(\Gamma \vdash \vec{B}_{i}^\otimes)$,
we obtain its \definitive{$i$-lifting}
$\lift_i(\isogQ) \in \IAug_\bullet(\Gamma \vdash o)$.
\begin{prop}\label{prop:def_base}
The previous construction on augmentations induces a bijection
\[
\lift : \sum_{1\leq i \leq n} \IAug(\Gamma \vdash \vec{B}_{i}^\otimes)
\bij
\IAug_\bullet(\Gamma \vdash o) \,.
\]
\end{prop}
\begin{proof}
    We set $\lift(i,\isogQ)=\ic{\lift_i(\rep{\isogQ})}$.

\emph{Surjective.}
Any $\augQ \in \Aug_\bullet(\Gamma \vdash o)$ has a unique initial
move, which is negative hence cannot be maximal by \emph{$+$-covered}.
By \emph{determinism}, there is a unique subsequent Player move,
displayed to the initial move of some $A_i$. The subsequent moves
define $\augQ' \in \Aug(\Gamma \vdash \vec{B}_{i}^\otimes)$
\emph{s.t.} $\augQ \iso \lift_i(\augQ')$.

\emph{Injective.} Given isomorphic $\augQ \iso \lift_i(\augQ')$ and
$\augP \iso \lift_j(\augP')$, we obviously have $i = j$ since
isomorphisms of augmentations preserve display maps; and the
isomorphism decomposes into $\augQ' \iso \augP'$ as required.
\end{proof}

Putting together the above results, we may now deduce:
\begin{thm}\label{th:bij_aug_term}
For $\Gamma$ a context and $\typeA$ a type, there are bijections:
\[
\begin{array}{lcrcl}
\sintr{-}_\Trm &:& \TrmOfNf{\Gamma}{\typeA} &\bij& \IAug_\bullet(\intr{\Gamma} \vdash \intr{\typeA})\\
\sintr{-}_\Bag &:& \BagOfNf{\Gamma}{\typeA} &\bij& \IAug(\intr{\Gamma} \vdash \intr{\typeA})\\
\sintr{-}_\Seq &:& \SeqOfNf{\Gamma}{\vtypeA} &\bij& \IAug(\intr{\Gamma} \vdash \intr{\vtypeA})\,.
\end{array}
\]
\end{thm}
\begin{proof}
The three functions are defined by mutual induction using 
Propositions~\ref{prop:def_seq}, \ref{prop:def_bag}, \ref{prop:def_val}
and \ref{prop:def_base}.
For injectivity, we reason directly by induction on the syntax,
using the injectivity of each construction.
For surjectivity, we reason by induction on the size
(\emph{i.e.}\ the number of events) of augmentations, the syntactic kind
(considering $\Trm<\Bag<\Seq$), and also the type $\typeA$ in the $\Trm$ case:
the decomposition provided by \autoref{prop:def_base} yields augmentations of
strictly smaller size (we remove the two initial moves);
the bijection of \autoref{prop:def_val} preserves the size of augmentations
and stays in the kind $\Trm$, but yields a smaller output type;
the remaining two decompositions do not increase the size and yield a lower kind.
\end{proof}

\section{A Category of Strategies}
\label{sec:strategies}

In the remaining of the paper, we will extend the previous static
correspondence between normal resource terms and isogmentations into a
\emph{dynamic} correspondence.
In the present section, we introduce \emph{strategies} as weighted sums of
isogmentations, and define a \emph{composition} of strategies so that 
strategies form a category \(\Strat\).
We will show in later sections that one can interpret (not necessarily normal)
resource terms as strategies, in a compositional way:
the interpretation is defined inductively on terms, each syntactic constructor
being reflected by a construction on strategies;
and it will follow from our results that composition is the image of resource
substitution through this interpretation, which is moreover invariant under
resource reduction.
We will expose the relevant structure for the interpretation by introducing 
the notion of \emph{resource category} in \autoref{sec:resource_categories}:
showing that \(\Strat\) is indeed a resource category, in
\autoref{sec:strategies_resource_cat}, will yield the correspondence announced as
the commutation of square $(c)$ in our introduction.

\subsection{Strategies and Composition}

Consider $A, B$ and $C$ three negative arenas, and fix two augmentations $\augQ
\in \Aug(A \vdash B)$, $\augP \in \Aug(B \vdash C)$. We shall compose them
via \emph{interaction}, followed by \emph{hiding}.

\paragraph{Interaction of augmentations.}
Intuitively, we can only compose $\augQ$ and $\augP$ provided they reach
the same state on $B$, so we first extract the ``state they reach'' via
their desequentializations: let us write $x^\augQ_A$ for the events of
$\augQ$ that display to $A$ and $x^\augQ_B$ for those that display to
$B$ -- these define\footnote{Up to isomorphism
of augmentations we may consider that $\deseq{\augQ} = x^\augQ_A \vdash
x^\augQ_B$, but we need not assume that.}
 $x^\augQ_A \in \conf(A)$ and $x^\augQ_B \in
\conf(B)$ and likewise for $\augP$.

But what does it mean to ``reach the same state''? In general requiring
$x^\augQ_B = x^\augP_B$ is meaningless, since this data should really
be considered up to isomorphism.
States in $B$ are not configurations, but \emph{positions}:
symmetry classes of configurations.
Thus $\augQ$ and $\augP$ are \definitive{compatible} if $x^\augQ_B$ and
$x^\augP_B$ are \definitive{symmetric}, \emph{i.e.}\ if there is
a symmetry $\varphi : x^\augQ_B \sym_B x^\augP_B$ --
we write $x^\augQ_B \sym_B x^\augP_B$ for the equivalence.
Accordingly, we must define the composition of two compatible
augmentations \emph{along with} a mediating symmetry. We
first form \emph{interactions}:

\begin{defi}\label{def:inter}
For $\augQ, \augP$ as above and $\varphi : x^\augQ_B \sym_B x^\augP_B$,
the \definitive{interaction} $\augP \inter_\varphi \augQ$ is the pair
$\tuple{
    \ev{\augP \inter_\varphi \augQ},
    {\leq}_{\augP \inter_\varphi \augQ}
}$ of the set $\ev{\augP \inter_\varphi \augQ}\eqdef\ev{\augQ}+\ev{\augP}$
and the binary relation ${\leq}_{\augP \inter_\varphi \augQ}$ on 
$\ev{\augP \inter_\varphi \augQ}$ defined as the transitive closure of
${\causes}\eqdef {\causes}_\augQ \cup {\causes}_\augP \cup {\causes}_\varphi$
with
\[
\begin{array}{rcl}
\causes_\augQ &=& \{((1, m), (1, m')) \mid m <_\augQ m'\}\,,\\
\causes_\augP &=& \{((2, m), (2, m')) \mid m <_\augP m'\}\,,\\
\causes_\varphi &=& 
\{((1, m), (2, \varphi(m)))
    \mid m \in x^\augQ_B \land \pol_{A\vdash B}(\display_{\augQ}(m)) = +\}\\ 
&\cup& \{((2, m), (1, \varphi^{-1}(m)))
    \mid m \in x^\augP_B \land \pol_{B \vdash C}(\display_\augP(m)) = +\}\,.
\end{array}
\]
\end{defi}

\begin{prop}\label{prop:inter_acyclic}
    The interaction $\augP \inter_\varphi \augQ$ is a partially ordered set.
\end{prop}
\begin{proof}
    We show that $\causes$ is \emph{acyclic}, \emph{i.e.}\ its transitive
    closure is a strict partial order.
    This corresponds to a subtle and fundamental property of innocent
    strategies: they always have a deadlock-free interaction.
    Our proof, which appears in \autoref{subsec:deadlockfree}, is a direct
    adaptation of a similar fact in concurrent games on event structures
    \cite{CC21}.
\end{proof}

\begin{figure}[t]
\[
    \begin{tikzcd}[gs header,column sep=0pt,row sep=.5cm]
        o&\tensor&o&\vdash&&o&&[1cm]&o&&\vdash(&o&\tto&o&\tto&o&)\tto&o
        \\[-.5cm+1ex]
        &&&&|[red]|\qu^-
        &&|[red]|\qu^-
        &&&&&&&&&&&\qu^-
        \\
        \qu^+\ar[urrrr,ddep,"\augQ" description]
        &&\qu^+\ar[urrrr,ddep,"\augQ" description]
        &&&&&&&&&&&&&\qu^+\ar[urr,ddep,"\augP" description]\ar[urr,sdep]
        \\
        &&&&&&&&&&&\qu^-\ar[urrrr,ddep,"\augP" description]\ar[urrrr,sdep]
        &&\qu^-\ar[urr,ddep,"\augP" description]\ar[urr,sdep]
        \\
        &&&&&&&|[red]|\qu^+\ar[urrrr,ddep,"\augP" description]
        \ar[uuulll,-,red,double,"\varphi" description]
        &&|[red]|\qu^+\ar[urrrr,ddep,"\augP" description]
        \ar[uuulll,-,red,double,"\varphi" description]
        \\
    \end{tikzcd}
\]
\caption{Construction of an interaction $\augP \inter_\varphi \augQ$.}
\label{fig:ex_interaction}
\end{figure}

Figure~\ref{fig:ex_interaction} illustrates the construction of an interaction.
The two augmentations $\augQ \in \Aug(o \tensor o \vdash o)$ -- on the left
hand side -- and $\augP \in \Aug(o \vdash (o \tto o \tto o) \tto o)$ -- on the
right hand side -- are shown with their common interface in red,
bridged by a symmetry $\varphi$.

In order to show that the forthcoming composition of augmentation is
well-behaved, we need to characterize immediate causal dependency in
the interaction. It turns out that this is very constrained -- this is
detailed in two lemmas: the first, for forward causality, follows.

\begin{lem}\label{lem:imc_int_forward}
If $m = (1, m') \imc_{\augP \inter_\varphi \augQ} n$ for $m' \in \ev{\augQ}$,
then we have:
\[
\begin{array}{rl}
\text{\emph{(1)}} &
\text{If $m'$ is negative in $\augQ$, then $n = (1, n')$ and $m'
\imc_\augQ n'$;}\\
\text{\emph{(2)}} &
\text{If $m'$ is positive in $\augQ$ and displaying to $A$, then $n =
(1, n')$ and $m' \imc_\augQ n'$;}\\ 
\text{\emph{(3)}} &
\text{If $m'$ is positive in $\augQ$ and displaying to $B$, then $n =
(2, \varphi(m'))$,} 
\end{array}
\]
and symmetrically for $m = (2, m') \imc_{\augP \inter_\varphi \augQ} n$
for $m' \in \ev{\augP}$.
\end{lem}
\begin{proof}
Necessarily, an immediate causal link must originate from one of the
clauses of the relation $\causes$ above. For \emph{(1)}, for polarity
reasons it can only be $(1, m') \causes_\augQ n$ so that $n = (1, n')$
with $m' <_\augQ n'$, and furthermore we must have $m' \imc_\augQ n'$ or
that would immediately contradict $m \imc_{\augP \inter_\varphi \augQ}
n$. For \emph{(2)}, similarly only the clause $\causes_\augQ$ may
apply. 

For \emph{(3)}, we must show that $m \causes_\augQ n$ is impossible. If
that was the case, then $n = (1, n')$ with $m' <_\augQ n'$, where we
must have $m' \imc_\augQ n'$ (or contradict $m \imc n$). But as $m'$ is
positive, by \emph{courtesy} we have $m' \imc_{\deseq{\augQ}} n'$, and
thus $\varphi(m') \imc_{\deseq{\augP}} \varphi(n')$ as $\varphi$ is an
order-isomorphism. And by \emph{rule-abiding}, that entails
$\varphi(m') <_\augP \varphi(n')$, so that altogether we have
\[
m \causes_\varphi (2, \varphi(m')) \causes_\augP (2, \varphi(n'))
\causes_\varphi n
\]
contradicting the fact that $m \imc_{\augP \inter_\varphi \augQ} n$ is
an immediate causal link.
\end{proof}

Symmetrically, in the ``backwards'' direction, we have:

\begin{lem}\label{lem:imc_int_backward}
If $m \imc_{\augP \inter_\varphi \augQ} (1, n') = n$ for $n' \in
\ev{\augQ}$,
then we have:
\[
\begin{array}{rl}
\text{\emph{(1)}} &
\text{If $n'$ is positive in $\augQ$, then $m = (1, m')$ and $m'
\imc_\augQ n'$;}\\
\text{\emph{(2)}} &
\text{If $n'$ is negative in $\augQ$ and displaying to $A$, then $m =
(1, m')$ and $m' \imc_\augQ n'$;}\\ 
\text{\emph{(3)}} &
\text{If $n'$ is negative in $\augQ$ and displaying to $B$, then $m =
(1, \varphi^{-1}(m'))$,} 
\end{array}
\]
and symmetrically for $m \imc_{\augP \inter_\varphi \augQ} (2, n')$
for $n' \in \ev{\augP}$.
\end{lem}
\begin{proof}
Analogous to the proof of \autoref{lem:imc_int_forward}.
\end{proof}

Now, we can go on and define the composition of augmentations.

\paragraph{Composing augmentations.} We compose $\augQ$ and $\augP$
\emph{via $\varphi$}, by \emph{hiding} the interaction.
Recall that we write
$\deseq{\augQ} \iso x^\augQ_A \vdash x^\augQ_B \in \conf(A \vdash B)$ and
$\deseq{\augP} \iso x^\augP_B \vdash x^\augP_C \in \conf(B \vdash C)$.

\begin{prop}
    The structure $\augP \odot_\varphi \augQ$ obtained by restricting
    $\augP \inter_\varphi \augQ$ to $x^\augQ_A + x^\augP_C$, with
\[
 \deseq{\augP \odot_\varphi \augQ}\eqdef x^\augQ_A \vdash x^\augP_C\,,
\qquad
\display_{\augP \odot_\varphi \augQ}((1,
m))\eqdef\display_{\augQ}((1,m))\,,
\qquad
\display_{\augP \odot_\varphi \augQ}((2, m))\eqdef\display_\augP((2,m))
\]
is an augmentation on $A \vdash C$.
\end{prop}
\begin{proof}
Let us call an event $m \in \ev{\augP \inter_\varphi \augQ}$
\textbf{visible} if it appears in $\augP \odot_\varphi \augQ$,
\textbf{hidden} otherwise.

\emph{Rule-abiding} is immediate from the definition and the fact that
$\augQ$ and $\augP$ are rule-abiding.
For \emph{negative}, by case analysis,
an event minimal in $\augP \odot_\varphi \augQ$ must also be minimal in
$\augP \inter_\varphi \augQ$,
it must thus come from $\augP$ and occur in $C$, and be
minimal in $\augP$, hence negative; \emph{$+$-covered} is symmetric.  

It remains to show that $\augP \odot_\varphi \augQ$ is a forest, and
that it satisfies conditions \emph{courteous} and \emph{deterministic}
-- for those, we show analogous properties on the interaction $\augP
\inter_\varphi \augQ$. First, it is a forest: any move can have at most
one causal predecessor by \autoref{lem:imc_int_backward}. Since
$\augQ$ and $\augP$ are deterministic, and $\causes_\varphi$ does not
branch, \autoref{lem:imc_int_forward} entails that $\augP
\inter_\varphi \augQ$ can only branch at negative visible events, from which it
follows that $\augP \odot_\varphi \augQ$ is deterministic. Finally if
$m \imc_{\augP \odot_\varphi \augQ} n$, this means we must have a
sequence in the interaction
\[
m \imc_{\augP \inter_\varphi \augQ} m_1 \imc_{\augP \inter_\varphi
\augQ} \dots \imc_{\augP \inter_\varphi \augQ} m_k \imc_{\augP
\inter_\varphi \augQ} n
\]
in $\augP \inter_\varphi \augQ$ where $m_1, \dots, m_k$ are hidden. But
then if $m$ is positive, $m_1$ cannot be hidden by
\autoref{lem:imc_int_forward}, so that $k=0$ and $m \imc_{\augP
\inter_\varphi \augQ} n$. But as both $m$ and $n$ are visible, this is
only possible if they both come from $\augQ$ or both come from $\augP$,
contradicting the courtesy of $\augQ$ or $\augP$. Likewise, if $n$ is
negative, we reason likewise relying on
\autoref{lem:imc_int_backward} instead.
\end{proof}

\begin{figure}[t]
\[
    \begin{tikzcd}[gs header,column sep=0pt,row sep=.5cm]
        o&\tensor&o&\vdash(&o&\tto&o&\tto&o&)\tto&o
        \\[-.5cm+1ex]
        &&&&&&&&&&\qu^-
        \\
        &&&&&&&&\qu^+\ar[urr,ddep]\ar[urr,sdep]
        \\
        &&&&\qu^-\ar[urrrr,ddep,]\ar[urrrr,sdep]
        &&\qu^-\ar[urr,ddep,]\ar[urr,sdep]
        \\
        \qu^+\ar[urrrr,ddep,]
        &&\qu^+\ar[urrrr,ddep,]
        \\
    \end{tikzcd}
\]
\caption{
    The composition $\augP \odot_\varphi \augQ$, with $\augQ$ and $\augP$ 
    from \autoref{fig:ex_interaction}.
}
\label{fig:ex_composition}
\end{figure}

The interaction in \autoref{fig:ex_interaction} yields
the augmentation in \autoref{fig:ex_composition}, the
\emph{composition of $\augQ$ and $\augP$ via $\varphi$}.
Composition is then extended to isogmentations:
for $\isogQ \in \IAug(A \vdash B)$ and $\isogP \in \IAug(B \vdash C)$,
we write $x^{\isogQ}_B\eqdef x^{\,\rep{\isogQ}}_B$
and $x^{\isogP}_B\eqdef x^{\,\rep{\isogP}}_B$, and then for
$\varphi : x^{\isogQ}_B \sym_B x^{\isogP}_B$, we define
$\isogP \odot_\varphi \isogQ \eqdef \ic{\rep{\isogP} \odot_\varphi \rep{\isogQ}}$.
We will soon show (in \autoref{subsec:catlaws}) that the composition of
augmentations is compatible with isomorphisms, so the previous lifting to
isogmentations does not depend on the choice of representatives.

One fact is puzzling: the composition of $\augQ$ and $\augP$ is only defined
once we have fixed a mediating $\varphi : x^\augQ_B \sym_B
x^\augP_B$, which is not unique -- for instance, writing $\qu\qu\in\Aug(o)$ for
the augmentation with two copies of the only move $\qu$ of $o$,
there are exactly two symmetries $\qu\qu \sym_o \qu\qu$.
Worse, the result of composition depends on the choice of $\varphi$:
if \autoref{fig:ex_interaction} was constructed with the symmetry
$\psi : \qu\qu \sym_o \qu\qu$ swapping the two moves, we would get the variant
$\augP \odot_\psi \augQ$ of \autoref{fig:ex_composition} with the two final
causal links crossed: $\isogP \odot_\varphi \isogQ$ and $\isogP \odot_\psi \isogQ$ 
are distinct, even up to iso.

This is reminiscent of the behaviour of resource substitution: e.g.,
\[
    \rsubst*{\labs{f}{\rappl{f}{\mset{x}\,\mset{x}}}}{x}{\mset{y,z}}
    = \labs{f}{\rappl{f}{\mset{y}\,\mset{z}}}
    + \labs{f}{\rappl{f}{\mset{z}\,\mset{y}}}\,.
\]
As substitution of resource terms yields \emph{sums} of resource terms, this
suggests that composition of isogmentations should produce \emph{sums} of
isogmentations, which we will call \emph{strategies}.

\paragraph{Strategies.} 
A \emph{strategy} is simply a weighted sum of isogmentations.

\begin{defi}\label{def:strategy}
A \definitive{strategy} on arena $A$ is a function
$\strS : \IAug(A) \to \rp$, 
where $\rp$ is the completed half-line of non-negative real numbers.
We then write $\strS : A$.
\end{defi}

We regard $\strS : A$ as a weighted sum
$\strS = \sum_{\isogQ \in \IAug(A)} \strS(\isogQ) \cdot \isogQ$,
and we write $\supp(\strS)$ for its support set:
$\supp(\strS)\eqdef\{\isogQ\in\IAug(A)\mid \strS(\isogQ)\not=0\}$.
We lift the composition of isogmentations to strategies via the formula
\begin{eqnarray}
\strT \odot \strS &\quad\eqdef\quad& \sum_{\isogQ \in \IAug(A \vdash B)}
\sum_{\isogP \in
\IAug(B \vdash C)} \sum_{\varphi : x^\isogQ_B \sym_B x^\isogP_B} \strS(\isogQ)
\strT(\isogP) \cdot (\isogP \odot_\varphi \isogQ)\label{eq:def_comp}
\end{eqnarray}
for $\strS : A \vdash B$
and $\strT : B \vdash C$, \emph{i.e.}\ $(\strT \odot \strS)(\isogR)$ is
the sum of $\strS(\isogQ) \strT(\isogP)$ over all triples $\isogQ, \isogP,
\varphi$ \emph{s.t.} $\isogR = \isogP \odot_\varphi \isogQ$ -- there are
no convergence issues, as we consider positive coefficients and we have been
careful to include $+\infty\in\rp$ in \autoref{def:strategy}.

\paragraph{Identities.}
We also introduce identities: \emph{copycat strategies}, formal sums of
specific isogmentations presenting typical copycat behaviour;
we start by defining their concrete representatives. 

Consider $x \in \conf(A)$ on negative arena $A$.
The augmentation $\augCC_x \in \Aug(A\vdash A)$,
called the \definitive{copycat} augmentation on $x$, has
desequentialization $\deseq{\augCC_x} \eqdef x \vdash x$
and causal order $x \vdash x$, augmented with
\[  
\begin{array}{rclcl}
(1, m) &\leq_{\augCC_x}& (2, n) &\qquad& \text{if $m \leq_x n$ and
$\pol_A(\display_x(m)) = +$,}\\ 
(2, m) &\leq_{\augCC_x}& (1, n) &\qquad& \text{if $m \leq_x n$ and
$\pol_A(\display_x(m)) = -$,}
\end{array}
\]
so $\augCC_x$ adds to $x \vdash x$ all
immediate causal links of the form $(2, m) \imc (1, m)$ for negative
$m$, and $(1, m) \imc (2, m)$ for positive $m$. Again, this lifts to
isogmentations by setting, for $\posX \in \pos(A)$, the \definitive{copycat
isogmentation} $\isogCC_\posX \in \IAug(A \vdash A)$ as the isomorphism
class of $\augCC_{\rep{\posX}}$.

The strategy $\id_A : A \vdash A$ should have the isogmentation
$\isogCC_\posX$ for all position $\posX \in \pos(A)$. But with which
coefficient?
If $\id_A$ is to be an identity, then there is no choice:
this coefficient must exactly cancel the sum over all symmetries in
\eqref{eq:def_comp}. Thus:
\begin{eqnarray}
    \id_A &\eqdef& \sum_{\posX \in \pos(A)} \frac{1}{\nSymOf{\posX}} \cdot \isogCC_\posX
\label{eq:def_cc}
\end{eqnarray}
where $\SymOf{\posX}$ is the group of \emph{endosymmetries} of $\posX$,
\emph{i.e.}\ of all $\varphi : \rep{\posX} \sym_A \rep{\posX}$ --
the cardinality of $\SymOf{\posX}$ does not depend on the choice of
$\rep{\posX}$.
This use of such a coefficient to compensate for future sums over sets of
permutations is reminiscent of the Taylor expansion of $\lambda$-terms
\cite{ER08}.

\subsection{Proof of the Categorical Laws}
\label{subsec:catlaws}

In this section, we show the main arguments behind the following
result:

\begin{thm}\label{th:strat_cat}
The negative arenas and strategies between them form a category, $\Strat$.
\end{thm}

This is proved in several stages. Firstly, we establish isomorphisms
corresponding to categorical laws, working concretely on augmentations
-- this means that these laws will refer to certain isomorphisms
explicitly. Then, we show that composition of augmentations is
compatible with isomorphisms, so that it carries out to isogmentations.
From all that, we are in position to conclude and prove that $\Strat$
is indeed a category.

\paragraph{Laws on the composition of augmentations.}
The following lemma specifies in what sense the copycat augmentation
is neutral for composition:

\begin{lem}[Neutrality]\label{lem:cc_neutral_aug}
    Consider $\augQ\in\Aug(A\vdash B)$, $x\in\conf(B)$ and
    $\varphi:x^\augQ_B \iso_B x$.
    Then, $\augCC_x \odot_\varphi \augQ \iso \augQ$.
    Likewise, for any $y \in \conf(A)$ and $\psi:y \iso_A x^\augQ_A$,
    we have $\augQ \odot_\psi \augCC_y \iso \augQ$.
\end{lem}
\begin{proof}
Recall that $\augCC_x \odot_\varphi \augQ$ is obtained by
considering $\augCC_x \inter_\varphi \augQ$ with events
$\ev{\augQ} + \ev{\augCC_x}$,
\emph{i.e.}\ $\ev{\augQ} + (x + x)$ with causal order as
described in \autoref{prop:inter_acyclic}. The composition
$\augCC_x \odot_\varphi \augQ$ is then the restriction to its visible events,
\emph{i.e.}\ $x^\augQ_A + (\emptyset + x)$. Then
\[ 
\ev{\augCC_x \odot_\varphi \augQ} = x^\augQ_A + (\emptyset + x)
\bij
x^\augQ_A + x
\stackrel{x^\augQ_A + \varphi^{-1}}{\bij}
x^\augQ_A + x^\augQ_B
\bij
\ev{\augQ}
\]
forms a bijection between the sets of events, which is checked to be an
isomorphism of augmentations by a direct analysis of the causal order
of $\augCC_x \odot_\varphi \augQ$.
\end{proof}

It may be surprising that $\augCC_x \odot_\varphi \augQ \iso \augQ$ regardless
of $\varphi$:
the choice of the symmetry is reflected in the induced isomorphism
$\phi_{\varphi}\colon \augCC_x \odot_\varphi \augQ \iso \augQ$,
which the statement of this lemma ignores. 
Similarly, we have:

\begin{lem}[Associativity]\label{lem:comp_associative_aug}
        Consider $\augQ \in \Aug(A \vdash B)$, $\augP \in \Aug(B \vdash C)$,
        $\augR \in \Aug(C \vdash D)$, and two symmetries
        $\varphi : x^\augQ_B \iso_B x^\augP_B$ and $\psi : x^\augP_C \iso_C x^\augR_C$.
        Then
        \[
        \augR \odot_{\psi'} \left( \augP \odot_\varphi \augQ \right)
        \iso
        \left( \augR \odot_{\psi} \augP  \right) \odot_{\varphi'} \augQ \,.
        \]
with $\varphi', \psi'$ obtained from $\varphi$ and $\psi$,
adjusting tags for disjoint unions in the obvious way.
\end{lem}
\begin{proof}
A routine proof, relating the two compositions to a ternary composition
$\augR \odot^3_\psi \augP \odot^3_\varphi \augQ : \Aug(A \vdash D)$, defined in
a way similar to binary composition.
\end{proof}

\paragraph{Congruence.}
To lift the previous results to isogmentations and strategies, composition
must preserve isomorphisms.

Consider augmentations $\augQ, \augQ' \in \Aug(A\vdash B)$ with
$\varphi : \augQ \iso \augQ'$, we know that $\varphi$ is an isomorphism of
configurations $\varphi : \deseq{\augQ} \iso_{A \vdash B} \deseq{\augQ'}$ --
a symmetry -- therefore it projects to symmetries
$\varphi_A : x^\augQ_A \iso_A x^{\augQ'}_A$ and
$\varphi_B : x^\augQ_B \iso_B x^{\augQ'}_B$.

\begin{lem}\label{lem:iso_comp}
    Consider $\augQ, \augQ' \in \Aug(A\vdash B)$,
    $\augP, \augP' \in \Aug(B \vdash C)$,
    and isomorphisms $\theta : x^\augQ_B \iso_B x^\augP_B$,
    $\theta' : x^{\augQ'}_B \iso_B x^{\augP'}_B$,
    $\varphi : \augQ \iso \augQ'$ and $\psi : \augP \iso \augP'$
    such that $\theta' \circ \varphi_B = \psi_B \circ \theta$.

Then, we have an isomorphism $\psi \odot_{\theta, \theta'}
\varphi : \augP \odot_\theta \augQ \iso \augP' \odot_{\theta'} \augQ'$. 
\end{lem}

The proof is a direct verification that the obvious morphism between
$\augP \odot_\theta \augQ$ and $\augP' \odot_{\theta'} \augQ'$ is indeed an
isomorphism.
The main consequence of this lemma is the following. Consider
$\augQ, \augQ' \in \Aug(A\vdash B)$, $\augP, \augP' \in \Aug(B \vdash C)$,
isomorphisms $\varphi : \augQ \iso \augQ'$ and $\psi : \augP \iso \augP'$,
not requiring any commutation property as above. Still, $\varphi$ and
$\psi$ project to symmetries
\[
\varphi_B : x^\augQ_B \sym_B x^{\augQ'}_B\,,
\qquad
\qquad
\psi_B : x^\augP_B \sym_B x^{\augP'}_B\,.
\]
inducing a bijection
\[
\begin{array}{rcrcl}
\chi &:& x^\augQ_B \sym_B x^\augP_B &\bij& x^{\augQ'}_B \sym_B x^{\augP'}_B\\
&& \theta &\mapsto& \psi_B \circ \theta \circ \varphi_B^{-1}\,,
\end{array}
\]
so that for any $\theta : x^\augQ_B \sym_B x^\augP_B$, we have $\augP
\odot_\theta \augQ \iso \augP' \odot_{\chi(\theta)} \augQ'$ by
\autoref{lem:iso_comp}.  It ensues that we can substitute one
representative for another when summing over all mediating symmetries: 
\[
\sum_{\theta : x^\augQ_B \sym_B x^\augP_B} \ic{\augP \odot_\theta \augQ} 
= \sum_{\theta : x^\augQ_B \sym_B x^\augP_B}
\ic{\augP' \odot_{\chi(\theta)} \augQ'}
= \sum_{\theta : x^{\augQ'}_B \sym_B x^{\augP'}_B}
\ic{\augP' \odot_{\theta} \augQ'}
\]
using the observation above, and re-indexing the sum following $\chi$ --
or in other words, the composition of strategies does not depend on the
choice of representative used for isogmentations. This is often
used silently throughout the development.

\paragraph{Categorical laws.} We are now equipped to show
\autoref{th:strat_cat}. First, the identity laws:

\begin{prop}\label{prop:cc_neutral}
        Consider $\strS : A \vdash B$. Then, $\id_B \odot \strS 
        = \strS \odot \id_A = \strS$.
\end{prop}
\begin{proof}
We focus on $\id_B \odot \strS$.
For any $\isogP \in \IAug(B \vdash B)$, we write
$\deseq{\rep{\isogP}} = x^\isogP_\lside \vdash x^\isogP_\rside$.
We have:
\begin{eqnarray*}
\id_B \odot \strS
        & = &
        \sum_{\isogQ \in \IAug(A \vdash B)}\sum_{\isogP \in \IAug(B \vdash
B)}
        \sum_{\varphi: x^\isogQ_B \iso_B x^\isogP_\lside} \strS(\isogQ)
\id_B(\isogP) \cdot
        \isogP \odot_\varphi \isogQ
        \\
        & = &
        \sum_{\isogQ \in \IAug(A \vdash B)}\sum_{\posX \in \pos(B)}
        \sum_{\varphi: x^\isogQ_B \iso_B \rep{\posX}}
        \frac{\strS(\isogQ)}{\nSymOf{\posX}} \cdot
        \isogCC_\posX \odot_\varphi \isogQ
\end{eqnarray*}
using the definitions of the composition and of the identity strategy.
Next, we compute
\begin{eqnarray*}
\id_B \odot \strS
        & = &
        \sum_{\isogQ \in \IAug(A \vdash B)}\sum_{\posX \in \pos(B)}
        \sum_{\varphi: x^\isogQ_B \iso_B \rep{\posX}}
        \frac{\strS(\isogQ)}{\nSymOf{\posX}} \cdot \isogQ
        \\
        & = &
        \sum_{\isogQ \in \IAug(A \vdash B)}
        \sum_{\varphi \in \Sym(\ic{x^{\isogQ}_B})}
        \frac{\strS(\isogQ)}{\nSymOf{\ic{x^{\isogQ}_B}}} \cdot \isogQ
        \\
        & = &
        \sum_{\isogQ \in \IAug(A \vdash B)}
        \strS(\isogQ)\cdot \isogQ
        \quad =\quad \strS
\end{eqnarray*}
by \autoref{lem:cc_neutral_aug} and direct
reasoning on symmetries.
The proof of the identity $\strS \odot \id_A = \strS$ is symmetric.  
\end{proof}

Notice how the sum over all symmetries exactly compensates for the
coefficient in \eqref{eq:def_cc}. Likewise, associativity of the
composition of strategies follows from
\autoref{lem:comp_associative_aug} and bilinearity of composition,
altogether concluding the proof that $\Strat$ is a category.

\section{Resource Categories}
\label{sec:resource_categories}

What further structure does $\Strat$ enjoy, that we could leverage to
design an interpretation of the resource calculus? It is
not clear how to turn $\Strat$ into any existing notion of differential
category. It is not cartesian, and thus is not a cartesian differential
category \cite{BCS09}.  Neither does it have the structure of a model of
differential linear logic; it is additive and symmetric monoidal but
there is no clear candidate for the exponential modality. Overall, it should be
regarded as a new categorical structure, whose morphisms are analogous
to (sums of) \emph{bags} of resource terms inhabiting a simple type,
rather than single resource terms.

Thus we develop \emph{resource categories}, models of the resource
calculus inspired by $\Strat$.

\subsection{Motivation and Definition}
\label{subsec:def_resource_cat}

Both substitution of resource terms and composition of
isogmentations generate \emph{sums}, so we need an additive structure.
Following \cite{BCLS20},
an \definitive{additive symmetric monoidal category (asmc)} is a symmetric
monoidal category where each homset is a commutative monoid, and each
operation preserves the additive structure.

\paragraph{Bialgebras.} As for differential categories, resource
categories build on \emph{bialgebras}:

\begin{defi}
Consider $\C$ an additive symmetric monoidal category.

A \definitive{bialgebra} on $\C$ is $(A,\delta_A, \epsilon_A,
\mu_A, \eta_A)$ with
$(A,\mu_A : A \tensor A \to A, \eta_A : I \to A)$
a commutative monoid (see \autoref{fig:alg_laws}), and 
$(A,\delta_A : A \to A \tensor A, \epsilon_A : A \to I)$
a commutative comonoid (satisfying dual identities),
with additional bialgebra laws (\autoref{fig:bialg_laws}). 
\end{defi}

\begin{figure}[t]
    \[
        \begin{tikzpicture}[rcat]
            \node at (0,0) (in1) {};
            \node at (1.6,0) (in2) {};
            \node at (2.4,0) (in3) {};
            \node[shape=circle, draw] at (0.8,-1) (mu1) {$\mu_A$};
            \node[shape=circle, draw] at (1.6, -2) (mu2) {$\mu_A$}; 
            \node at (1.6, -3) (out1) {};
            \draw (mu2) -- (out1);
            \draw (in1) .. controls (0,-1) .. (mu1);
            \draw (in2) .. controls (1.6, -1) .. (mu1);
            \draw (mu1) .. controls (0.8, -2) .. (mu2);
            \draw (in3) .. controls (2.4, -2) .. (mu2);
            
            \node at (2.9, -1) (=1) {$=$};
            
            \node at (3.4,0) (in4) {};
            \node at (4.2,0) (in5) {};
            \node at (5.8,0) (in6) {};
            \node[shape=circle, draw] at (5,-1) (mu3) {$\mu_A$};
            \node[shape=circle, draw] at (4.2, -2) (mu4) {$\mu_A$}; 
            \node at (4.2, -3) (out2) {};
            \draw (mu4) -- (out2);
            \draw (in4) .. controls (3.4,-2) .. (mu4);
            \draw (mu3) .. controls (5,-2) .. (mu4);
            \draw (in5) .. controls (4.2,-1) .. (mu3);
            \draw (in6) .. controls (5.8,-1) .. (mu3);
        \end{tikzpicture}
        \qquad
        \begin{tikzpicture}[rcat]
           \node[shape=circle, draw] at (0,-1) (eta1) {$\eta_A$};
           \node at (2,0) (in1) {};
           \node[shape=circle, draw] at (1,-2) (mu1) {$\mu_A$};
           \node at (1,-3) (out1) {};      
           \draw (eta1) .. controls (0,-2) .. (mu1);
           \draw (in1) .. controls (2,-2) .. (mu1);
           \draw (mu1) -- (out1);
           
           \node at (2.5,-1) (=1) {$=$};
           
           \node at (3,0) (in2) {};
           \node at (3,-3) (out2) {};
           \draw (in2) -- (out2);
           
           \node at (3.5, -1) (=2) {$=$};
           
           \node[shape=circle, draw] at (6,-1) (eta2) {$\eta_A$};
           \node at (4,0) (in3) {};
           \node[shape=circle, draw] at (5,-2) (mu2) {$\mu_A$};
           \node at (5,-3) (out3) {};      
           \draw (in3) .. controls (4,-2) .. (mu2);
           \draw (eta2) .. controls (6,-2) .. (mu2);
           \draw (mu2) -- (out3);
        \end{tikzpicture}
        \qquad
        \begin{tikzpicture}[rcat]
            \node at (0,0) (in1) {};
            \node at (2,0) (in2) {};
            \node[shape=circle, draw] at (1,-2) (mu1) {$\mu_A$}; 
            \node at (1, -3) (out1) {};
            \draw (mu1) -- (out1);
            \draw (in1) .. controls (0, -2) .. (mu1);
            \draw (in2) .. controls (2, -2) .. (mu1);
            
            \node at (2.5, -1) (=) {$=$};
            
            \node at (3,0) (in3) {};
            \node at (5,0) (in4) {};
            \node[shape=circle, draw] at (4,-2) (mu2) {$\mu_A$}; 
            \node at (4, -3) (out2) {};
            \draw (mu2) -- (out2);
            \draw (in3) .. controls (3, -1) and (4,-1) .. (4,-1) .. controls (5,-1) and (5,-2) .. (mu2);
            \draw (in4) .. controls (5,-1) and (4,-1) .. (4, -1) .. controls (3,-1) and (3, -2) .. (mu2);
        \end{tikzpicture}
    \]
    \caption{Monoid laws}
    \label{fig:alg_laws}
\end{figure}

\begin{figure}
    \[
        \begin{tikzpicture}[rcat]
            \node at (0,0) (in1) {};
            \node at (1.6,0) (in2) {};
            \node[shape=circle, draw] at (0.8,-1) (mu1) {$\mu_A$};
            \node[shape=circle, draw] at (0.8,-2.5) (delta1) {$\delta_A$};
            \node at (0,-3.5) (out1) {};
            \node at (1.6,-3.5) (out2) {};
            \draw (in1) .. controls (0,-1) .. (mu1);
            \draw (in2) .. controls (1.6,-1) .. (mu1);
            \draw (mu1) -- (delta1);
            \draw (delta1) .. controls (0,-2.5) .. (out1);
            \draw (delta1) .. controls (1.6,-2.5) .. (out2);
            
            \node at (2.1, -1.5) (=) {$=$};
            
            \node at (3.4,0) (in3) {};
            \node at (5,0) (in4) {};
            \node[shape=circle, draw] at (3.4,-1) (delta2) {$\delta_A$};
            \node[shape=circle, draw] at (5,-1) (delta3) {$\delta_A$};
            \node[shape=circle, draw] at (3.4,-2.5) (mu2) {$\mu_A$};
            \node[shape=circle, draw] at (5,-2.5) (mu3) {$\mu_A$};
            \node at (3.4,-3.5) (out3) {};
            \node at (5,-3.5) (out4) {};
            
            \draw (in3) -- (delta2);
            \draw (in4) -- (delta3);
            \draw (delta2) .. controls (2.6,-1) and (2.6,-2.5) .. (mu2);    
            \draw (delta2) .. controls (4.2,-1) and (4.2,-2.5) .. (mu3);    
            \draw (delta3) .. controls (4.2,-1) and (4.2,-2.5) .. (mu2);
            \draw (delta3) .. controls (5.8,-1) and (5.8,-2.5) .. (mu3);
            \draw (mu2) -- (out3);
            \draw (mu3) -- (out4);
        \end{tikzpicture}
        \qquad
        \begin{tikzpicture}[rcat]
            \node[shape=circle, draw] at (0,0) (eta1) {$\eta_A$};
            \node[shape=circle, draw] at (0,-1.5) (delta1) {$\delta_A$};
            \node at (-0.8,-2.5) (out1) {};
            \node at (0.8,-2.5) (out2) {};
            \draw (eta1)--(delta1);
            \draw (delta1) .. controls (-0.8, -1.5) .. (out1) ;
            \draw (delta1) .. controls (0.8,-1.5) .. (out2) ;
            
            \node at (1.3,-1) (=) {$=$};
            
            \node[shape=circle, draw] at (2,0) (eta2) {$\eta_A$};
            \node[shape=circle, draw] at (3,0) (eta3) {$\eta_A$};
            \node at (2,-2.5) (out3) {};
            \node at (3,-2.5) (out4) {};
            \draw (eta2) -- (out3);
            \draw (eta3) -- (out4);
        \end{tikzpicture}
        \qquad
        \begin{tikzpicture}[rcat]
            \node at (-0.8,0) (in1) {};
            \node at (0.8,0) (in2) {};
            \node[shape=circle, draw] at (0,-1) (mu1) {$\mu_A$};
            \node[shape=circle, draw] at (0,-2.3) (epsilon1) {$\epsilon_A$};
            \draw (mu1)--(epsilon1);
            \draw (in1) .. controls (-0.8, -1) .. (mu1) ;
            \draw (in2) .. controls (0.8,-1) .. (mu1) ;
            
            \node at (1.3,-1) (=) {$=$};
            
            \node[shape=circle, draw] at (2,-2.3) (epsilon2) {$\epsilon_A$};
            \node[shape=circle, draw] at (3,-2.3) (epsilon3) {$\epsilon_A$};
            \node at (2,0) (in3) {};
            \node at (3,0) (in4) {};
            \draw (in3) -- (epsilon2);
            \draw (in4) -- (epsilon3);
        \end{tikzpicture}
        \qquad
        \begin{tikzpicture}[rcat]
            \node[shape=circle, draw] at (0,0) (eta) {$\eta_A$};
            \node[shape=circle, draw] at (0,-2) (epsilon) {$\epsilon_A$};
            \draw (eta) -- (epsilon);
            \node at (1, -1) (=) {$=$};
        \end{tikzpicture}
    \]
    \caption{Bialgebra laws}
    \label{fig:bialg_laws}
\end{figure}

In a resource category, all objects shall be bialgebras.
Comonoids $(A, \delta_A, \epsilon_A)$ are the usual categorical description of
duplicable objects.
Intuitively, requests made to $\delta_A$ on either side of the
tensor on the rhs, are sent to the lhs. Categorically the monoid
structure $(A, \mu_A, \eta_A)$ is dual, but its intuitive behaviour is
different: each request on the rhs is forwarded, non-deterministically,
to either side of the tensor on the lhs, reflecting the \emph{sums}
arising in substitutions.
These operational intuitions will be captured precisely by the resource
category structure of $\Strat$.

In contrast with differential categories, morphisms in a resource category
intuitively correspond to (sums of) \emph{bags} rather than (sums of)
\emph{terms}.
Morally, the empty bag from $A$ to $B$ is captured from
the bialgebra structure as $\eta_B \circ \epsilon_A \in \C(A, B)$,
written $1$.  Likewise, the \definitive{product}
$f * g = \mu_B \circ (f \tensor g) \circ \delta_A \in \C(A, B)$
of $f, g \in \C(A, B)$
captures the union of bags. This makes $(\C(A,
B), *, 1)$ a commutative monoid, altogether turning 
$\C(A, B)$ into a commutative semiring, although composition and tensor
in $\C$ only preserve the additive monoid.

A bag of morphisms may be \definitive{flattened} into a morphism by the
following operation: if $\bag{f} = [f_1, \dots, f_n] \in \Mf(\C(A,
B))$, we write $\Pi \bag{f} \eqdef f_1 * \dots * f_n \in \C(A, B)$.

\paragraph{Pointed identities.}
Resource categories axiomatize categorically the \emph{singleton bags}.
For that, a pivotal role is played by the \definitive{pointed identity}, a
chosen idempotent $\pid_A \in \C(A, A)$ which we think of as a
singleton bag with a linear copycat behaviour.
From this we may capture the singleton bags
as those $f \in \C(A, B)$ such that $\pid_B \circ f = f$.
Dually, we may also capture the morphisms that access \emph{exactly one}
resource: those are the $f \in \C(A, B)$ such that $f \circ \pid_A = f$.
More formally:

\begin{defi}
Consider $\C$ an asmc where each object is equipped with a bialgebra
structure.
For $A \in \C$, a \definitive{pointed identity} on $A$ is an idempotent
$\pid_A \in \C(A, A)$
satisfying the equations shown as string diagrams in
\autoref{fig:laws_pointed}, plus $\epsilon_A \circ \pid_A = 0$ and
$\pid_A \circ \eta_A = 0$.
\end{defi}

\begin{figure}[t]
\[
\begin{tikzpicture}[rcat,scale=.6,baseline=(base)]
\node[style={circle,draw}] (pid) at (0, 0) {$\pid_A$};
\node[style={circle,draw}] (delta) at (0, -2) {$\delta_A$};
\draw(pid) edge (delta);
\draw(delta) edge[bend right=40] (-2,-4);
\draw(delta) edge[bend left=40] (2, -4);
\draw(0, 2) edge (pid);
\coordinate (base) at (0, -1);
\end{tikzpicture}
=
\quad
\begin{tikzpicture}[rcat,scale=.6,baseline=(base)]
\node[style={circle,draw}] (pid) at (0, 0) {$\pid_A$};
\node[style={circle,draw}] (eta) at (2, 0) {$\eta_A$};
\draw(pid) edge (0, -3);
\draw(0, 3) edge (pid);
\draw(eta) edge (2, -3);
\coordinate (base) at (0, 0);
\end{tikzpicture}
\quad
+
\quad
\begin{tikzpicture}[rcat,scale=.6,baseline=(base)]
\node[style={circle,draw}] (pid) at (2, 0) {$\pid_A$};
\node[style={circle,draw}] (eta) at (0, 0) {$\eta_A$};
\draw (pid) edge (2, -3);
\draw (2, 3) edge (pid);
\draw (eta) edge (0, -3);
\coordinate (base) at (0, 0);
\end{tikzpicture}
\qquad
\begin{tikzpicture}[rcat,scale=.6,baseline=(base)]
\node[style={circle,draw}] (pid) at (0, 0) {$\pid_A$};
\node[style={circle,draw}] (mu) at (0, 2) {$\mu_A$};
\draw (mu) edge (pid);
\draw (-2, 4) edge[bend right=40] (mu);
\draw (2, 4) edge[bend left=40] (mu);
\draw (pid) edge (0,-2);
\coordinate (base) at (0, 1);
\end{tikzpicture}
=
\quad
\begin{tikzpicture}[rcat,scale=.6,baseline=(base)]
\node[style={circle,draw}] (pid) at (0, 0) {$\pid_A$};
\node[style={circle,draw}] (epsilon) at (2, 0) {$\epsilon_A$};
\draw (pid) edge (0, -3);
\draw (0, 3) edge (pid);
\draw (2, 3) edge (epsilon);
\coordinate (base) at (0, 0);
\end{tikzpicture}
\quad
+
\quad
\begin{tikzpicture}[rcat,scale=.6,baseline=(base)]
\node[style={circle,draw}] (pid) at (2, 0) {$\pid_A$};
\node[style={circle,draw}] (epsilon) at (0, 0) {$\epsilon_A$};
\draw(pid) edge (2, -3);
\draw(2, 3) edge (pid);
\draw(0, 3) edge (epsilon);
\coordinate (base) at (0, 0);
\end{tikzpicture}
\]
\caption{Laws for (co)multiplication and pointed identity}
\label{fig:laws_pointed}
\end{figure}

Those laws are reminiscent of the laws of derelictions and
coderelictions in bialgebra modalities
\cite{BCLS20}, except that both roles are played
by $\pid_A$. 
In a resource category $\C$, all objects have a pointed identity.  The
``singleton bags'' are those $f \in \C(A, B)$ that are
\definitive{pointed}, \emph{i.e.}\ $\pid_B \circ f = f$ -- we write
$f\in\C_\bullet(A, B)$. Dually, $f \in \C(A, B)$ is
\definitive{co-pointed} if $f \circ \pid_A = f$, and we write $f \in
\C^\bullet(A, B)$.

\paragraph{Resource categories.}
Altogether, we are now ready to define resource categories:

\begin{figure}[t]
    \begin{gather*}
        \begin{tikzcd}[column sep=-1cm,ampersand replacement=\&]
        \&A\tensor B
        \ar[rr,"\delta_{A\tensor B}"]
        \ar[dl,"\delta_A\tensor\delta_B"']
        \&\&(A\tensor B)\tensor(A\tensor B)
        \ar[dr,"\alpha_{A\tensor B,A,B}"]
        \\
        (A\tensor A)\tensor (B\tensor B)
        \ar[d,"\alpha_{A\tensor A,B,B}"']
        \&\&\&\&((A\tensor B)\tensor A)\tensor B
        \ar[d,"\alpha^{-1}_{A,B,A}\tensor B"]
        \\
        ((A\tensor A)\tensor B)\tensor B
        \ar[drr,"\alpha^{-1}_{A,A,B}\tensor B"']
        \&\&\&\&(A\tensor (B\tensor A))\tensor B
        \\
        \&\&(A\tensor (A\tensor B))\tensor B
        \ar[urr,"(A\tensor\gamma_{A,B})\tensor B"']
    \end{tikzcd}
    \\[1em]
    \begin{tikzcd}[ampersand replacement=\&]
        A\tensor B
        \ar[r,"\epsilon_{A\tensor B}"]
        \ar[d,"\epsilon_A\tensor \epsilon_B"']
        \&I
        \\
        I\tensor I\ar[ur,"\lambda_I=\rho_I"']
    \end{tikzcd}
    \qquad
    \begin{tikzcd}[ampersand replacement=\&]
        I
        \ar[r,bend left,"\epsilon_I"]
        \ar[r,bend right,"\id_I"']
        \&I
    \end{tikzcd}
    \end{gather*}
\caption{Compatibility of comonoids with the monoidal structure --
there are symmetric conditions for the compatibility of monoids with
the monoidal structure.}  
\label{fig:compat_com_mon}
\end{figure}

\begin{defi}
A \definitive{resource category} is an asmc $\C$ where each $A \in \C$ has a
bialgebra structure $(A, \delta_A, \epsilon_A, \mu_A, \eta_A)$ and pointed
identity $\pid_A$, such that the bialgebra structure is compatible with the
monoidal structure of $\C$ in the sense that the diagrams of
\autoref{fig:compat_com_mon} commute, and symmetrically for the monoid structure,
with symmetric coherence laws.
\end{defi}

This simple definition has powerful consequences. In particular, the
following key property, derived from the definition of resource
categories, expresses how the product of a bag of pointed morphisms
interacts with the comonoid structure -- and dually for the product of a
bag of co-pointed morphisms and monoids. Much of the forthcoming proof
that the interpretation is invariant under reduction relies on it:

\begin{figure}[t]
    \[
        \begin{tikzcd}
            A
            \ar[r,"\delta_A"]
            \ar[d,"\Pi \bag{f}"']
            & A\tensor A 
            \ar[d,"\sum_{\bag{f} \splitinto \bag{f}_1 * \bag{f}_2} \Pi \bag{f}_1 \tensor \Pi \bag{f}_2"]
            \\
            B
            \ar[r,"\delta_B"']
            &
            B\tensor B
        \end{tikzcd}
        \qquad
        \begin{tikzcd}
            A\tensor A
            \ar[r,"\mu_A"]
            \ar[d,"\sum_{\bag{f} \splitinto \bag{f}_1 * \bag{f}_2} \Pi \bag{f}_1 \tensor \Pi \bag{f}_2"']
            & A
            \ar[d,"\Pi \bag{f}"]
            \\
            B\tensor B
            \ar[r,"\mu_B"']
            &
            B
        \end{tikzcd}
    \]
    \caption{Compatibility of bags of pointed (resp.\ co-pointed) morphisms with $\delta$ (resp.\ $\mu$).}
    \label{fig:compat_bags}
\end{figure}

\pagebreak[5]

\begin{restatable}{lem}{keylemma}\label{lem:main_lemma}
Consider $\C$ a resource category, then we have the following
properties:
\begin{enumerate}[(1)]
    \item \label{lem:main_lemma:1}
        For any bag of pointed morphisms
        $\bag{f} \in \Mf(\C_\bullet(A, B))$, 
        \begin{enumerate}[(a)]
            \item the first diagram of \autoref{fig:compat_bags}
                commutes; and 
            \item we have $\epsilon_B \circ \Pi \bag{f} = 1$ if $\bag{f}$
                is empty, $0$ otherwise;
        \end{enumerate}
    \item For any bag of co-pointed morphisms
        $\bag{f} \in \Mf(\C^\bullet(A, B))$,
        \begin{enumerate}[(a)]
            \item the second diagram of \autoref{fig:compat_bags}
                commutes; and 
            \item we have $\Pi\bag{f} \circ \eta_A = 1$ if $\bag{f}$ is
                empty, $0$ otherwise.
        \end{enumerate}
\end{enumerate}
\end{restatable}
\begin{proof}
This follows from a lengthy but mostly direct diagram chase.
\end{proof}

\paragraph{Closed resource categories.}
The notion of resource category is enough to reflect categorically the
mechanism of resource substitution, but in order to also capture resource
\emph{reduction}, we must furthermore impose some closure condition.

First, assuming that the monoidal structure of $\C$ is closed,
\emph{i.e.}\ that the endofunctor $- \tensor A$ has
a right adjoint $A \tto -$ for each $A \in \C$, we obtain the
following constructions, borrowing the terminology and notations
from the theory of cartesian closed categories:
if $f \in \C(A \tensor B, C)$,
the \definitive{currying of $f$} is the image of $f$ by the adjunction,
written
\[\Lambda_{A, B, C}(f) \in \C(A, B \tto C)\,,\]
and for any two objects $A, B \in \C$, we have
the \definitive{evaluation morphism}
\[
       \evm_{A, B} 
       \eqdef \Lambda^{-1}_{A\tto B, A, B}
       (\id_{A\tto B})  \in \C((A \tto B) \tensor A, B)\,.
\]

\begin{defi}\label{defi:closed}
    A resource category is said to be \definitive{closed} if
    its monoidal structure is closed and, moreover,
    currying is compatible with pointed identities:
    for all $A,B \in \C$, 
    \[ 
        \pid_{A \tto B} = \Lambda_{A \tto B,A,B}(\pid_{B} \circ \evm_{A,B})\,.
    \]
\end{defi}

\subsection{Interpretation of the Resource Calculus}

In order to describe the interpretation of the resource calculus, it
will be convenient to introduce some additional combinators,
mimicking the structure of cartesian categories in resource categories:

\paragraph{Cartesian combinators.}
The \definitive{pairing} of $f \in \C(\Gamma, A)$ and $g \in \C(\Gamma, B)$ is 
\[
\tuple{f, g} \eqdef (f \tensor g) \circ \delta_\Gamma \in \C(\Gamma, A \tensor B)\,;
\]
likewise
$\pi_1 \eqdef \rho_A \circ (A \tensor \epsilon_B) \in \C(A \tensor B, A)$ and
$\pi_2 \eqdef \lambda_B \circ (\epsilon_A \tensor B) \in \C(A \tensor B, B)$
are the two \definitive{projections}
-- we shall also use their obvious $n$-ary generalizations.
The laws of cartesian categories fail in general:
we have $\tuple{\pi_1, \pi_2} = \id_{A\tensor B}$, but
\emph{e.g.} $\pi_1 \circ \tuple{f, h} = f$ only holds if $h$ is
\definitive{erasable} (\emph{i.e.}\ $\epsilon_B \circ h = \epsilon_\Gamma$)
and $\tuple{f, g} \circ h = \tuple{f \circ h, g \circ h}$ if $h$ is
\definitive{duplicable}
(\emph{i.e.}\ $\delta_\Gamma \circ h = (h \tensor h) \circ \delta_\Delta$)
-- so we do get the usual laws if $h$ is a \emph{comonoid morphism}
\cite{M09}.

\paragraph{Lemmas on propagation of substitutions.}
Morphisms coming from the interpretation are not comonoid morphisms,
but many structural morphisms are: for instance it follows from a
direct diagram chase that projections \emph{are} comonoid morphisms.

As explained above, comonoid morphisms propagate in tuples as in a
cartesian category.
But, importantly, resource categories also specify how some non comonoid
morphisms propagate through a pairing, even paired with a comonoid morphism: 

\begin{lem}\label{lem:tuple_comp_prod}
Let $\C$ be a resource category,
$\bag{b} \in \Mf(\C_\bullet(\Delta, A))$, $h \in
\C(\Delta, \Gamma)$, $f \in \C(\Gamma \tensor A,
B)$, $g \in \C(\Gamma \tensor A, C)$.
If $h$ is a comonoid morphism, then we have:
\[  
\tuple{f, g} \circ \tuple{h, \Pi \bag{b}} = \sum_{\bag{b} \splitinto
\bag{b}_1 * \bag{b}_2} \tuple{f \circ \tuple{h, \Pi
\bag{b}_1}, g \circ \tuple{h, \Pi \bag{b}_2}}\,.
\]

If moreover $B=C$ then
\[
(f * g) \circ \tuple{h,\Pi \bag{b}}
= \sum_{\bag{b} \splitinto \bag{b}_1 * \bag{b}_2}
(f \circ \tuple{h, \Pi{\bag{b}_1}}) * (g \circ \tuple{h, \Pi{\bag{b}_2}}))\,.
\]

Finally, $1 \circ \tuple{h, \Pi{\bag{b}}} = 1$ if $\bag{b}$ is empty, $0$ otherwise.
\end{lem}
\begin{proof}
This is a simple verification based on \autoref{lem:main_lemma:1} of
\autoref{lem:main_lemma}. 

For the first equality, for instance, we start by unfolding:
\begin{eqnarray}
\tuple{f, g} \circ \tuple{h, \Pi \bag{b}} = (f \tensor g) \circ
\delta_{\Gamma \tensor A} \circ (h \tensor \Pi \bag{b}) \circ
\delta_\Delta
\label{eq:aux1}
\end{eqnarray}
and then we reason equationally, denoting by
$i_{A, B} : (A \tensor A) \tensor (B \tensor B) \iso (A \tensor B)
\tensor (A \tensor B)$ the obvious structural isomorphism as in
\autoref{fig:compat_com_mon}:
\begin{eqnarray*}
\delta_{\Gamma \tensor A} \circ (h \tensor \Pi \bag{b}) 
&=& i_{\Gamma, A} \circ (\delta_\Gamma \tensor \delta_A) \circ (h \tensor
\Pi \bag{b})\\
&=& i_{\Gamma, A} \circ (\delta_\Gamma \circ h) \tensor (\delta_A \circ
\Pi \bag{b})\\
&=& i_{\Gamma, A} \circ ((h \tensor h) \circ \delta_\Delta) \tensor
((\sum_{\bag{b} \splitinto \bag{b}_1 * \bag{b}_2} \Pi \bag{b}_1 \tensor
\Pi \bag{b}_2) \circ \delta_\Delta)\\
&=& \sum_{\bag{b} \splitinto \bag{b}_1 * \bag{b}_2} i_{\Gamma, A} \circ
((h \tensor h) \circ \delta_\Delta) \tensor ((\Pi \bag{b}_1 \tensor
\Pi \bag{b}_2) \circ \delta_\Delta)\\
&=& \sum_{\bag{b} \splitinto \bag{b}_1 * \bag{b}_2} i_{\Gamma, A} \circ
((h \tensor h) \tensor (\Pi \bag{b}_1 \tensor
\Pi \bag{b}_2)) \circ (\delta_\Delta \tensor \delta_\Delta)\\
&=& \sum_{\bag{b} \splitinto \bag{b}_1 * \bag{b}_2}
((h \tensor \Pi \bag{b}_1) \tensor (h \tensor \Pi \bag{b}_2)) \circ
i_{\Delta, \Delta} \circ (\delta_\Delta \tensor \delta_\Delta)\\
&=& \sum_{\bag{b} \splitinto \bag{b}_1 * \bag{b}_2}
((h \tensor \Pi \bag{b}_1) \tensor (h \tensor \Pi \bag{b}_2)) \circ
\delta_{\Delta \tensor \Delta}
\end{eqnarray*}
using compatibility of $\delta_A$ with the monoidal structure,
bifunctoriality of the tensor, the fact that $h$ is a comonoid morphism
along with \autoref{lem:main_lemma:1} of \autoref{lem:main_lemma},
linearity of all operations with respect to the sum, again
bifunctoriality of the tensor, naturality of $i_{A, B}$, and again
compatibility of $\delta$ with the monoidal structure.

Substituting this in \eqref{eq:aux1}, we obtain (accounting for linearity):
\begin{eqnarray*}
\tuple{f, g} \circ \tuple{h, \Pi \bag{b}} &=& 
\sum_{\bag{b} \splitinto \bag{b}_1 * \bag{b}_2}
(f \tensor g) \circ ((h \tensor \Pi \bag{b}_1) \tensor (h \tensor \Pi
\bag{b}_2)) \circ
\delta_{\Delta \tensor \Delta} \circ \delta_\Delta\\
&=&\sum_{\bag{b} \splitinto \bag{b}_1 * \bag{b}_2}
(f \tensor g) \circ ((h \tensor \Pi \bag{b}_1) \tensor (h \tensor \Pi
\bag{b}_2)) \circ (\delta_\Delta \tensor \delta_\Delta) \circ
\delta_\Delta\\
&=& \sum_{\bag{b} \splitinto \bag{b}_1 * \bag{b}_2}
\tuple{f \circ \tuple{h, \Pi \bag{b}_1}, g \circ \tuple{h, \Pi
\bag{b}_2}}
\end{eqnarray*}
where $\delta_\Delta \circ \delta_{\Delta \tensor \Delta} =
\delta_\Delta \circ (\delta_\Delta \circ \delta_\Delta)$ follows from
the fact that $\Delta$ is a \emph{commutative} comonoid, and the rest
is by definition of tupling, and bifunctoriality of tensor.

The second equality has essentially the same proof, and the last is
straightforward.
\end{proof}

This is fairly close to how substitutions propagate through terms in
the resource $\lambda$-calculus (see \autoref{sec:syntax}): we sum
over all the partitions of the bag $\bag{b}$ into two components, to be
distributed to the two components of the pair -- when using this lemma
in the proof of the substitution lemma,  the comonoid morphism $h$
shall simply be an identity leaving all the unsubstituted variables
unchanged. Syntactic substitution has another important case, namely
when a substitution encounters a variable occurrence. Likewise here, we have:

\begin{restatable}{lem}{pickone}\label{lem:pickone}
Consider $\bag{f} \in \Mf(\C_\bullet(A, B))$.
Then $\pid_B \circ \Pi \bag{f} = v$ if $\bag{f} = [v]$, $0$ otherwise.
\end{restatable}

Again this lemma follows by diagram chasing, using the conditions of a resource
category.
It illustrates how the pointed identity is able to pick a single element of a
bag.
If the bag has too many elements or not enough, then the composition yields $0$
-- this reflects syntactic substitution, where the substitution must offer
\emph{exactly} as many resources as there are occurrences of the
substituted variable.

\paragraph{Interpretation.}
From now on, we fix a closed resource category $\C$ with a chosen
object $o$.

We first set
$\intr{\typeo} \eqdef o$,
$\intr{\tuple{\typeA_1,\dotsc,\typeA_n}} \eqdef
\intr{\typeA_1}\otimes\cdots\otimes\intr{\typeA_n}$
and $\intr{\typeA \to \typeB} \eqdef \intr{\typeA} \tto \intr{\typeB}$.
For contexts,
$\intr{\Gamma} \eqdef \bigotimes_{(x : \typeA) \in \Gamma} \intr{\typeA}$.
Note that currying and associativity morphisms induce an isomorphism
\[\begin{tikzcd}
    \intr{\typeA} \ar[r,"\zeta_{\typeA}"]
    &
    \intr{\vtypeB}\tto\typeo
\end{tikzcd}\]
for any type $\typeA=\vtypeB\to\typeo$.
If $(x:\typeA)\in\Gamma$, we then write
\[\begin{tikzcd}
    \intr{\Gamma} \ar[r,"\var^\Gamma_x"]
    &
    \intr{\vtypeB}\tto\typeo
\end{tikzcd}\]
for the projection morphism $\intr{\Gamma}\to\intr{\typeA}$
followed by $\zeta_{\typeA}$.
For $\Gamma$ and $\Delta$ disjoint we also use the iso
$
\merge_{\Gamma, \Delta} \in \C(\intr{\Gamma} \otimes \intr{\Delta},
\intr{\Gamma, \Delta})
$.
The properties of closed resource categories ensure the following
commutation:
\begin{lem}
  For any type $\typeA=\vtypeB\to\typeo$,
  \( 
      \zeta_{\typeA}\circ\pid_{\intr{\typeA}} =
      \pid_{\intr{\vtypeB\to\typeo}} \circ \zeta_{\typeA}
      \,.
  \)
\end{lem}

The interpretation of terms (or, rather, of typing derivations) then
follows the three kinds of judgements from \autoref{sec:syntax}:
for $\Gamma, A \in \C$ and $\vec{A} = \tuple{A_1, \dots, A_n}$, we define
\[  
\Trm_\C(\Gamma; A) \eqdef \C_\bullet(\Gamma, A)
\quad
\Bag_\C(\Gamma; A) \eqdef \Mf(\Trm_\C(\Gamma; A))
\quad
\Seq_\C(\Gamma; \vec{A}) \eqdef
\Pi_{1\leq i \leq n} \Bag_\C(\Gamma; A_i)
\,.
\]

Notably, sequences and bags are interpreted as \emph{actual} sequences and
bags at the ``meta-level'', rather than via the ``internal'' bags
(\emph{i.e.}\ products of pointed maps) or products (\emph{i.e.}\ via the
monoidal structure) in $\C$.
This apparent duplication of structure will be resolved when interpreting
applications.
For that purpose, in addition to the product $\Pi\bag{f}\in\C(\Gamma,A)$
of a bag of morphisms $\bag{f}\in\Bag_\C(\Gamma;A)$, we also define the
\definitive{packing}
\[
    \pack{\seq{f}\,} \eqdef
    \tuple{\Pi \bag{f}_1, \dots, \Pi \bag{f}_n} \in \C(\Gamma,\vec{A}^\otimes)
\]
of any sequence of morphisms
$\vec{f}= \tuple{\bag{f}_1, \dots, \bag{f}_n}\in\Seq_\C(\Gamma;\vec{A})$.

We now define the three interpretation functions
\[
\Trm(\Gamma; \typeA) \to \Trm_\C(\intr{\Gamma}; \intr{\typeA})
\quad
\Bag(\Gamma; \typeA) \to \Bag_\C(\intr{\Gamma}; \intr{\typeA})
\quad
\Seq(\Gamma; \vtypeA) \to \Seq_\C(\intr{\Gamma}; \intr{\vtypeA})
\]
all written $\intr{-}$, by mutual induction,
as in \autoref{fig:interpretation}.
The interpretation is extended to sums of terms
$\STrm(\Gamma; \typeA) \to \Trm_\C(\intr{\Gamma}; \intr{\typeA})$
relying on the additive structure of $\C$ -- we give no interpretation to
sums of bags or sequences.

\begin{figure}[t]
    \begin{align*}
        \intr{\Gamma \jugTrm \labs{\varA}{\termA}:\typeA \to \typeB}
        &\eqdef
        \Lambda_{\intr{\Gamma}, \intr{\typeA}, \intr{\typeB}}
        (\intr{\Gamma, \varA:\typeA \jugTrm \termA:\typeB}
        \circ \merge_{\intr{\Gamma}, \intr{x : \typeA}})
        \\ 
        \intr{\Gamma \jugTrm \rappl{\varA}{\vbagB}:\typeo}
        &\eqdef
        \evm_{\intr{\vtypeA},o} \circ
        \tuple{
            \pid_{\intr{\vtypeA}\tto o} \circ \var^{\Gamma}_{\varA},
            \pack{\intr{\Gamma \jugSeq \vbagB :\vtypeA}}
        }\\
        \intr{\Gamma \jugTrm \rappl{\termA}{\bagB}:\typeB}
        &\eqdef
        \evm_{\intr{\typeA},\intr{\typeB}} \circ \tuple{
            \intr{\Gamma \jugTrm \termA :\typeA\to\typeB},
            \Pi{\intr{\Gamma \jugBag \bagB :\typeA}}
        }\\
        \intr{\Gamma\jugBag\mset{\termA_1,\dotsc,\termA_n}:\typeA}
        &\eqdef
        \mset{
            \intr{\Gamma\jugTrm\termA_i:\typeA} \mid 1 \leq i \leq n
        }\\
        \intr{\Gamma\jugSeq\tuple{\bagA_1,\dotsc,\bagA_n}:\vtypeA}
        &\eqdef
        \tuple{
            \intr{\Gamma\jugBag\bagA_i:\typeA_i} \mid 1 \leq i \leq n
        }
    \end{align*}
\caption{Interpretation of the resource calculus}
\label{fig:interpretation}
\end{figure}

\subsection{The Soundness Theorem}

We show that this interpretation is invariant under reduction.

\paragraph{Semantic substitution.}
The bulk of the proof consists in
proving a suitable substitution lemma, for which we must first give a
semantic account of substitution.
We define three semantic substitution functions:
\[
\begin{array}{rcrcl}
\ssubst{-}{\vec{x}}{-} 
&:& \Trm_\C(\intr{\Gamma,\varA:\typeA}; \intr{\typeB})
\times \Bag_\C(\intr{\Gamma}; \intr{\typeA})
&\to& \Trm_\C(\intr{\Gamma}; \intr{\typeB})\\
\ssubst{-}{\vec{x}}{-}
&:& \Bag_\C(\intr{\Gamma,\varA:\typeA};\intr{\typeB})
\times \Bag_\C(\intr{\Gamma}; \intr{\typeA})
&\to& \C(\intr{\Gamma},\intr{\typeB})\\
\ssubst{-}{\vec{x}}{-}
&:& \Seq_\C(\intr{\Gamma,\varA:\typeA}; \intr{\vtypeB})
\times \Bag_\C(\intr{\Gamma}; \intr{\typeA})
&\to& \C(\intr{\Gamma}, \intr{\vtypeB})
\end{array}
\]
using our cartesian-like notations:
\begin{align*}
\ssubst{f}{\varA}{\bag{g}}
& \eqdef f \circ \merge_{\intr{\Gamma}, \intr{\varA : \typeA}}
\circ \tuple{\id_{\intr{\Gamma}}, \Pi{\bag{g}}}
\\
\ssubst{\bag{f}}{\varA}{\bag{g}}
&\eqdef \Pi{\bag{f}} \circ
\merge_{\intr{\Gamma}, \intr{\varA : \typeA}} \circ
\tuple{\id_{\intr{\Gamma}}, \Pi{\bag{g}}}
\\
\ssubst{\seq{f}}{\varA}{\bag{g}}
&\eqdef \pack{\seq{f}\,} \circ
\merge_{\intr{\Gamma}, \intr{\varA : \typeA}} \circ
\tuple{\id_{\intr{\Gamma}}, \Pi{\bag{g}}}
\end{align*}
for
$f \in \Trm_\C(\intr{\Gamma,\varA:\typeA}; \intr{\typeB})$,
$\bag{f} \in \Bag_\C(\intr{\Gamma,\varA:\typeA}; \intr{\typeB})$,
and $\seq{f} \in \Seq_\C(\intr{\Gamma, \varA:\typeA}; \intr{\vtypeB})$
respectively.
We may now state the substitution lemma:

\begin{restatable}{lem}{substitutionlemma}\label{lem:substitution_lemma}
Consider $\bagB \in \BagOf{\Gamma}{\typeA}$,
$\Delta = \Gamma, \varA : \typeA$ and $\termA \in \TrmOf{\Delta}{\typeB}$.
Then, 
\[ 
    \intr{\rsubst{\termA}{\varA}{\bagB}} 
    = \ssubst{\intr{\termA}}{\varA}{\intr{\bagB}}\,.
\]
\end{restatable}
\begin{proof}
We show the stronger statement that for all
$\bagB \in \BagOf{\Gamma}{\typeA}$
and $\Delta = \Gamma, \varA : \typeA$,
\[
\begin{array}{rl}
\text{\emph{(1)}}
& \text{if $\termA \in \TrmOf{\Delta}{\typeB}$, then
$\intr{\rsubst{\termA}{\varA}{\bagB}} =
\ssubst{\intr{\termA}}{\varA}{\intr{\bagB}}$;}\\
\text{\emph{(2)}}
& \text{assume $\bagA \in \BagOf{\Delta}{\typeB}$ with
    $\rsubst{\bagA}{\varA}{\bagB} = \sum_{1\leq i \leq n}
\bagA_i$, where $\bagA_i \in \BagOf{\Gamma}{\typeB}$,}\\
&\text{then $\sum_{1\leq i \leq n} \Pi{\intr{\bagA_i}} =
\ssubst{\intr{\bagA}}{\varA}{\intr{\bagB}}$;}\\
\text{\emph{(3)}}
& \text{assume $\vbagA \in \SeqOf{\Delta}{\vtypeB}$ with
    $\rsubst{\vbagA}{\varA}{\bagB} = \sum_{1\leq i \leq n}
\vbagA_i$, where  $\vbagA_i \in \SeqOf{\Gamma}{\vtypeB}$,}\\
&\text{then $\sum_{1\leq i \leq n} \pack{\intr{\vbagA_i}} =
\ssubst{\intr{\vbagA}}{\varA}{\intr{\bagB}}$;}
\end{array}
\]
which follows by induction on typing derivations, using the lemmas of the
previous subsection.
\end{proof}

From the substitution lemma above, we may easily
deduce:

\begin{restatable}{thm}{invariance}
If $\termsA\in\STrmOf{\Gamma}{\typeA}$ and $\termsA\bred\termsA'$ then
$\intr{\termsA} = \intr{\termsA'}$.
\end{restatable}
\begin{proof}
Preservation of $\beta$-reduction follows from \autoref{lem:substitution_lemma}. 
Invariance for bags and sequences must be stated carefully: reduction
yields sums of bags and sequences whereas the sets $\Bag_\C(\Gamma, A)$
and $\Seq_\C(\Gamma, \seq{A})$ are not stable under sums.
To show that invariance extends by context closure, we prove the three
statements:
\[
\begin{array}{rl}
\text{\emph{(1)}} &
\text{if $\termA\in\TrmOf{\Gamma}{\typeA}$ and $\termA\bred\termsA'$
then $\intr{\termA} = \intr{\termsA'}$;}\\
\text{\emph{(2)}} &
\text{if $\bagA\in\BagOf{\Gamma}{\typeA}$ and
    $\bagA\bred\sum_{i\in I} \bagA_i$
    then $\Pi{\intr{\bagA}} = \sum_{i\in I} \Pi{\intr{\bagA_i}}$;}\\
\text{\emph{(3)}} &
\text{if $\vbagA\in\SeqOf{\Gamma}{\vtypeA}$ and
    $\vbagA\bred\sum_{i\in I} \vbagA_i$
    then $\pack{\intr{\vbagA}} = \sum_{i\in I} \pack{\intr{\vbagA_i}}$;}
\end{array}
\]
by mutual induction, following the inductive definition of context
closure. Finally, it is immediate that this extends to sums as
required.
\end{proof}

\section{Game Semantics as a Resource Category}
\label{sec:strategies_resource_cat}

It remains to check that $\Strat$ is indeed a resource category, and
that the induced interpretation of normal forms coincides with the
bijections from \autoref{th:bij_aug_term}. 

\subsection{Additive Symmetric Monoidal Structure}

\paragraph{Tensor.} As for composition we first define
the tensor of augmentations, then isogmentations, then strategies.
For $A_i$, $B_i$ arenas with $\augQ_i \in
\Aug(A_i \vdash B_i)$ for $i=1,2$, we set $\augQ_1 \tensor \augQ_2 \in
\Aug(A_1 \tensor A_2 \vdash B_1 \tensor B_2)$
with
$\ev{\augQ_1 \tensor \augQ_2} = \ev{\augQ_1} + \ev{\augQ_2}$ and 
$\display_{\augQ_1 \tensor \augQ_2}(i, m) = (j,(i,n))
\,\text{ if }\, \display_{\augQ_i}(m) = (j,n)$,
and the orders $\leq_{\augQ_1 \tensor \augQ_2}$ and $\leq_{\deseq{\augQ_1
\tensor \augQ_2}}$ inherited. This construction preserves isomorphisms,
hence the tensor $\isogQ_1 \tensor \isogQ_2 \in \IAug(A_1 \tensor A_2 \vdash
B_1 \tensor B_2)$ may be defined via any representative -- for
definiteness, we use the chosen representatives of $\isogQ_1$ and $\isogQ_2$.
We lift the definition to strategies with,
for $\strS_1 : \Gamma_1 \vdash A_1$ and $\strS_2 : \Gamma_2 \vdash
A_2$:
\[
\strS_1 \tensor \strS_2 \eqdef
\sum_{\isogQ_1 \in \IAug(A_1 \vdash B_1)} \,
\sum_{\isogQ_2 \in \IAug(A_2 \vdash B_2)} \,
\strS_1(\isogQ_1) \, \strS_2(\isogQ_2) \, \cdot
\left( \isogQ_1 \tensor \isogQ_2 \right) \,.
\]

\paragraph{Structural morphisms.} Structural morphisms are all
variations of copycat. As we did for copycat itself, we start
with concrete representatives. Consider $A$, $B$, $C$ arenas, and
$x \in \conf(A)$, $y \in \conf(B)$, $z \in \conf(C)$. Noting $\vide$ the empty
configuration on $1$, we set:
\[
\begin{array}{lclclcl}
        \deseq{\lambda_A^x} & = & \vide \tensor x \vdash x \,,
        & \qquad &
        \deseq{ \alpha_{A,B,C}^{x,y,z}} & = &
        x \tensor (y \tensor z) \vdash (x \tensor y) \tensor z \,,
        \\
        \deseq{\rho_A^x} & = & x \tensor \vide \vdash x \,,
        & \qquad &
        \deseq{\gamma_{A,B}^{x,y}} & = & x \tensor y \vdash y \tensor x
\,.
        
\end{array}
\]
and $\lambda_A^x$, $\rho_A^x$, $\alpha_{A,B,C}^{x,y,z}$ and
$\gamma_{A,B}^{x,y}$
are defined from these, augmented with the obvious copycat behaviour.

We lift this to isogmentations: for $\posX \in \pos(A)$,
$\blambda_A^\posX$ is the isomorphism class of $\lambda_A^{\rep{\posX}}$; and
likewise for the others. Then the strategy $\lambda_A$ is defined as
for $\id_A$ in \eqref{eq:def_cc} (\autopageref{eq:def_cc}) by setting
\[
    \lambda_A \eqdef \sum_{\posX \in \pos(A)}
    \frac{1}{\nSymOf{\posX}} \cdot \blambda_A^{{\posX}}
\]
and likewise for $\rho_A$, $\alpha_{A,B,C}$ and $\gamma_{A,B}$.
These
structural morphisms satisfy the necessary conditions to make $(\Strat,
\tensor, 1)$ a symmetric monoidal category.

\paragraph{Additive Structure.}
The additive structure is simply given by the formal sum.
Consider $\Gamma$, $A$ negative arenas and $\strS, \strT : \Gamma \vdash A$,
then the sum $\strS + \strT : \Gamma \vdash A$ is:
\[
       \strS + \strT \eqdef \sum_{\isogQ \in \IAug(\Gamma \vdash A)}
       \left( \strS(\isogQ) + \strT(\isogQ) \right) \cdot \isogQ \,, 
\]
and $0$ is the empty strategy ($\supp(0)=\emptyset$). 
The tensor and the composition are compatible with the additive
structure, hence $\Strat$ is an asmc.

\subsection{Resource Category Structure}

\paragraph{Bialgebra.}
As resource categories rest on bialgebras, this is our next step in
understanding the categorical structure of $\Strat$. 
For the strategies for (co)multiplication, we first set configurations
$\deseq{\delta_A^{x,y}} = x * y \vdash x \tensor y$ and
$\deseq{\mu_A^{x,y}} = x \tensor y \vdash x * y$
for any $A$ and $x, y \in \conf(A)$; $\delta_A^{x,y}$, $\mu_A^{x,y}$ are
obtained by adjoining copycat behaviour on $x$ and $y$. This lifts to
isogmentations $\bdelta_A^{\posX,\posY}$ and $\bmu_A^{\posX,\posY}$ for
$\posX,\posY\in\pos(A)$.
We can now define the strategies as:
\[
       \delta_A \eqdef
       \sum_{\posX,\posY \in \pos(A)}
       \frac{1}{\nSymOf{\posX} \times \nSymOf{\posY}}
       \cdot \bdelta_A^{\posX,\posY}
       \qquad
       \mu_A \eqdef
       \sum_{\posX,\posY \in \pos(A)}
       \frac{1}{\nSymOf{\posX} \times \nSymOf{\posY}}
       \cdot \bmu_A^{\posX,\posY}
       \,.
\]
The unit $\epsilon_A$ and co-unit and $\eta_A$ are both strategies
with only the empty isogmentation in their support, with coefficient
$1$. This provides the components for the following result:

\begin{prop}
For any negative arena $A$,
the tuple $(A,\delta_A, \epsilon_A, \mu_A, \eta_A)$ is a bialgebra.
\end{prop}
\begin{proof}
Most of the identities involved in the bialgebra structure are
straightforward. The only subtle law is the distributivity between
$\delta_A$ and $\mu_A$, which is detailed in \autoref{app:proof_bialgebra}.
\end{proof}

Additionally, it is direct that this is compatible with the monoidal structure.

\paragraph{Pointed Identity.}
Now, we focus on the next component of resource categories:
pointed identities and morphisms.
So we start by defining the pointed identity as the copycat strategy
over pointed positions only:
\[
       \pid_A = \sum_{\posX \in \pos_\bullet(A)}
       \frac{1}{\nSymOf{\posX}} \cdot \isogCC_{\posX}
\]
where $\pos_\bullet(A)$ denotes pointed positions of $A$,
\emph{i.e.}\ $\posX\in\pos(A)$ such that $\rep{\posX}$ is pointed.
Note that for any $\posX \in \pos_\bullet(A)$, 
we have $\deseq{\rep{\isogCC_{\posX}}} = \rep{\posX} \vdash \rep{\posX}$,
where both $\rep{\posX}$'s are pointed.
Hence, isogmentations in $\id_A^\bullet$ have both left and right 
components pointed
-- which explains why $\pid_A$ satisfies the pointed identity laws.

This definition corresponds to the intuition given in the previous section.
We want \emph{pointed morphisms} to represent singleton bags
-- strategies whose isogmentations are all pointed.
But for any $\strS:\Gamma \vdash A$, since $\pid_A$ is defined over pointed
positions only, the composition $\pid_A \odot \strS$ preserves only the
isogmentations in $\strS$ which are themselves pointed 
-- hence $\pid_A \odot \strS = \strS$ iff all isogmentations of $\strS$ are
pointed.
Dually, we want \emph{co-pointed morphisms} to represent morphisms which
require exactly one resource. Consider again $\strS : \Gamma \vdash A$, 
then $\strS \odot \pid_\Gamma$ preserves only the isogmentations
which have their left component pointed -- hence $\strS$ is 
co-pointed iff the left components of all of its isogmentations are pointed.

\paragraph{Closed structure.}
Finally, we look at monoidal closure.
We use the currying bijection $\Lambda_{\Gamma, A, B}$
from \autoref{sec:resource_as_aug}. For
$\strS : \Gamma \tensor A \vdash B$, we set
$        \Lambda_{\Gamma, A, B}(\strS)
        =
        \sum_{\isogQ \in \IAug(\Gamma \tensor A \vdash B)}
        \strS(\isogQ) \cdot \Lambda_{\Gamma, A, B}(\isogQ)$,
which directly yields 
\[
\Lambda_{\Gamma, A, B}: \Strat(\Gamma \tensor A, B) \iso \Strat(\Gamma, A \tto B)
\,.
\]
The laws of monoidal closure are easily
verified, as well as the additional identity of
\autoref{defi:closed}, finally yielding the desired structure:
\begin{restatable}{thm}{thmStrat}
        $\Strat$ is a closed resource category.
\end{restatable}

\subsection{Compatibility with normal forms}
\label{subsec:compat_nf}
Finally, we show that, up to the bijection $\sintr{-}$
between normal resource terms and isogmentations,
the interpretation of a resource term in the resource category $\Strat$
coincides with its normal form.

\begin{prop}\label{prop:compat_nf}
Consider $\termA \in \TrmOfNf{\Gamma}{\typeA}$.
Then $\intr{\termA}$ is the sum having $\sintr{\termA}$ with
coefficient $1$, and $0$ everywhere else.
\end{prop}
\begin{proof}
Recall the interpretation of the resource calculus in a resource category 
given in \autoref{fig:interpretation}.
Restricting to normal forms rules out the third clause in that figure.
We treat the remaining cases inductively.
The identity follows immediately from the induction hypothesis
for sequences, bags and abstraction terms, since
the interpretation matches the bijections of Propositions
\ref{prop:def_seq}, \ref{prop:def_bag} and \ref{prop:def_val}.
The case of a fully applied variable is less obvious, but the $i$-lifting
described in \autoref{fig:lifting} actually corresponds to the composition
with $\evm_{\intr{\vtypeA}, o}$ in the variable clause of
\autoref{fig:interpretation},
recalling that
\[
       \evm_{\intr{\vtypeA}, o} 
       = \Lambda^{-1}_{\intr{\vtypeA}\tto o, \intr{\vtypeA}, o}
       (\id_{\intr{\vtypeA} \tto o}) \,.
       \qedhere
\]
\end{proof}

Thus, our interpretation in $\Strat$ computes a representation of the
normal form:

\begin{cor}\label{cor:main}
If $\termA \in \TrmOf{\Gamma}{\typeA}$ has normal form $\sum_{i\in I}
\termA_i$, then $\intr{\termA} = \sum_{i\in I} \sintr{\termA_i}$.
\end{cor}

\section{Concluding remarks}
\label{sec:conclusion}

The correspondence with game semantics relies on the terms of the resource
calculus to be $\eta$-expanded. This was expected
-- as in \cite{TO16} --
but some consequences deserve discussion.

Firstly, $\varA : \typeA \to \typeB$ is not a valid term as it is not
$\eta$-long: it hides some infinitary copycat behaviour, requiring an
infinite sum as in \eqref{eq:def_cc}. This makes our calculus finitary
in a stronger sense than usual: each normal resource term describes a
simple, finite behaviour, and one can prove that it corresponds to a single
point of the relational model -- as in \cite{TO16} again.
This also means that in the absence of infinite sums, our typed
syntax is \emph{not} a resource category as it lacks identities.

Secondly, one might think that having an $\eta$-long syntax puts the pure,
untyped $\lambda$-calculus out of reach.
In \cite{TO16}, Tsukada and Ong suggest the resource calculus with tests
\cite{BCEM12} as a candidate for extending the
correspondence to the untyped setting, but this does not seem fit for the task:
it gives a syntactic counterpart to points of the relational model that do not
correspond to any normal resource term nor any pointed augmentation.

It is in fact possible to enforce full $η$-expansion on terms without typing,
but this requires altering the syntax of the calculus allowing for infinite
sequences of abstractions, as well as applications to infinite sequences of
(almost always empty) bags.
Keeping the correspondence between terms and augmentations then relies on
finding the analogue of a reflexive object in the category of games.
And the resulting \emph{extensional resource calculus}, enjoys a close
relationship with Nakajima trees \cite{N75}
(see also \cite[Exercise 19.4.4]{B85}),
similar to that of the ordinary resource calculus with Böhm trees.
More precisely, it is possible to design an extensional version 
of Taylor expansion, so that the normal form of the Taylor expansion
of a $λ$-term characterizes exactly its Nakajima tree or, equivalently, its
class in the greatest sensible consistent $λ$-theory $\mathcal H^*$.
This is the subject of a separate paper, in preparation \cite{BCVA24}.

\appendix

\section{Proof of the absence of deadlocks in interactions}
\label{subsec:deadlockfree}

The goal of this section is to prove \autoref{prop:inter_acyclic}:
fixing $\augQ \in \Aug(A\vdash B)$ and $\augP \in \Aug(B \vdash C)$ and
$\varphi : x^\augQ_B \sym_B x^\augP_B$, we must prove that $\causes$ is acyclic
on $\augP \inter_\varphi \augQ$.

If $m \in \ev{\augP \inter_\varphi \augQ}$, it \textbf{occurs in $A$}
if it has form $(1, m')$ with $\partial_\augQ(m') = (1, a)$; and
\textbf{occurs in $C$} if it has form $(2, m')$ with
$\partial_\augP(m') = (2, c)$. Otherwise, it \textbf{occurs in $B$}.

First, we note that a cycle in $\causes$ induces a cycle that is entirely
in $B$.

\begin{lem}\label{lem:cycle_inB}
Consider $\augQ \in \Aug(A \vdash B)$, $\augP \in \Aug(B \vdash C)$
and $\varphi : x^\augQ_B \iso_B x^\augP_B$.

If $\causes$ has a cycle in $\augP \inter_{\varphi} \augQ$,
then it has a cycle entirely in $B$.
\end{lem}
\begin{proof}
First, observe that $\causes$ has no direct link between $A$
and $C$. Consider a cycle
\[
m_1 \causes \dots \causes m_n \causes m_1\,,
\]
first note that this cycle must pass through $B$. Otherwise, it is
entirely in $A$ or entirely in $C$, making $\causes_\augQ$ or
$\causes_\augQ$ cyclic, contradiction. Now, consider a section
\[
m_i^B \causes m_{i+1}^A \causes \dots m_j^A
\causes m_{j+1}^B
\]
with the segment $m_{i+1} \causes \dots \causes m_j$
entirely in $A$. By definition, we must have
\[
m_i^B \causes_\augQ m_{i+1}^A \causes_\augQ \dots
\causes_\augQ m_j^A \causes_\augQ m_{j+1}^B\,,
\]
so that $m_i^B \causes_\augQ m_{j+1}^B$ by transitivity of
$\causes_\augQ$. Hence, a segment of the cycle in $A$ may be
removed, preserving the cycle. Symmetrically, any segment in $C$ may be
removed, yielding a cycle within $B$.
\end{proof}

Hence, we restrict our attention to cycles entirely within $B$.
Given a cycle
\[
m_1 \causes \dots \causes m_n \causes m_1
\]
we call $n$ its \textbf{length}. Next we show that this cycle can also
be assumed not to contain any element minimal in $B$. We
introduce additional conventions. If $m$ occurs in $B$, we say it
\textbf{occurs in $B^\lside$} if it is $(1, m')$ for $m' \in \ev{\augQ}$,
and \textbf{occurs in $B^\rside$} otherwise.

We will use this lemma:

\begin{lem}\label{lem:min_static_dynamic}
        Consider $q\in \Aug(A \vdash B)$ and $m, m' \in \ev{\augQ}$ such that $m
        <_\augQ m'$, with $m'$ occurring in $B$ and $\partial_\augQ(m)$ minimal in $B$.
        Then, $m <_{\deseq{\augQ}} m'$.
\end{lem}
\begin{proof}
        Since $\deseq{\augQ}$ is a forest, there is a unique $n <_{\deseq{\augQ}} m'$
        such that $n$ is minimal in $B$. By \emph{rule-abiding}, $n$ is minimal
        for $\leq_\augQ$. Since $\augQ$ is a forest, it follows that $n = m$.
\end{proof}

Exploiting that, we prove:

\begin{lem}\label{lem:cycle_nf}
        Consider $\augQ \in \Aug(A\vdash B)$, $\augP \in\Aug(B \vdash C)$
        and $\varphi : x^\augQ_B \iso_B x^\augP_B$.

        If $\causes$ has a cycle, then it has one entirely in $B$ and
        without minimal event in $B$.
\end{lem}
\begin{proof}
        By \autoref{lem:cycle_inB}, assume the cycle is entirely in $B$.
        Consider a cycle
        \[
        m_1 \causes \dots \causes m_n \causes m_1
        \]
        entirely in $B$ of minimal length.
        Seeking a contradiction, consider $m_i$ minimal in $B$.

        Assume first it is in $B^\rside$.
        Since $B$ is negative, $m_i$ is positive in $\augP$,
        hence we cannot have 
        $m_{i-1} \causes_\varphi m_i$ and we must have 
        $m_{i-1} \causes_\augP m_i$.
        Then, if $m_{i+1}$ is also in $B^\rside$, we have
        \[
        m_{i-1} \causes_\augP m_i \causes_\augP m_{i+1}
        \]
        so $m_{i-1} \causes_\augP m_{i+1}$, shortening the cycle and
        contradicting its minimality. So $m_{i+1}$ is in $B^\lside$. But then
        $m_i \causes_\varphi m_{i+1}$ so that $\varphi(m_{i+1})
        = m_i$. As an isomorphism of
        configurations, $\varphi$ preserves the display to $B$, so $m_{i+1}$ is
        minimal in $B^\lside$.

        In all cases, the cycle thus contains a minimal element of $B^\lside$,
        call it $m_i=(1,m'_i)$, with $m'_i$ negative in $\augQ$:
        since $\augQ$ is courteous, $m'_i$ is also minimal in $\augQ$
        so we cannot have $m_{i-1} \causes_\augQ m_i$.
        Then
        \[
        m_{i-1} \causes_\varphi
        m_i \causes_\augQ m_{i+1} \causes_\augQ \dots
        \causes_\augQ m_j \causes_\varphi m_{j+1}
        \]
        where all relations in between $m_i$ and $m_j=(1,m'_j)$ are in
        $\causes_\augQ$ (by definition, only those can apply until we
        jump to $B^\rside$ via $\causes_\varphi$), and where by
        definition, $\varphi(m'_j) = m'_{j+1}$.
        But then, by transitivity, $m_i \causes_\augQ m_j$,
        \emph{i.e.}\ $m'_i<_\augQ m'_j$.
        But by \autoref{lem:min_static_dynamic},
        this entails $m'_i <_{\deseq{\augQ}} m'_j$.
        Now as $\varphi$ is an isomorphism of configurations,
        this entails $\varphi(m'_i) <_{\deseq{\augP}} \varphi(m'_j)$,
        and by \emph{rule-abiding}, this entails
        $\varphi(m'_i) <_{\augP} \varphi(m'_j)$,
        \emph{i.e.}\ $m_{i-1} \causes_\augP m_{j+1}$.
        But this means that the segment $m_i \dots m_j$ may be
        removed from the cycle, contradicting the minimality of the latter.
\end{proof}

So it suffices to focus on cycles entirely in $B$, comprising no
minimal event.

Let us write $\deseq{\augP \inter_\varphi \augQ}$ for
$\ev{\augQ}+\ev{\augP}$, partially ordered by:
$(i,m)<_{\deseq{\augP \inter_\varphi \augQ}}(j,n)$
iff $i=j=1$ and $m <_{\deseq\augQ} n$, 
 or $i=j=2$ and $m <_{\deseq\augP} n$. 
If $m \in \ev{\augP \inter_\varphi \augQ}$ is not minimal in $\deseq{\augP
\inter_\varphi \augQ}$, it has a unique predecessor in $\deseq{\augP
\inter_\varphi \augQ}$ called its \textbf{justifier} and written
$\just(m)$.
We say $m$ occurring in $B$ has \textbf{polarity $\lside$} if it has the form
$(1, m')$ with $\pol_{A\vdash B}(\partial_\augQ(m')) = +$,
\textbf{has polarity $\rside$} if it has the form $(2, m')$ with
$\pol_{B \vdash C}(\partial_\augP(m')) = +$, and
\textbf{has polarity $\varphi$} otherwise.
We may then write $m^{\lside}$, $m^{\rside}$ or $m^{\varphi}$
instead of $m$, depending on its polarity.

\begin{lem}\label{lem:deadlock_free_aux}
        We have the following properties:
        \[  
        \begin{array}{rl}
                \text{\emph{(1)}} &
                \text{if $m \causes_\augQ n^\varphi$, then $m \causes_\augQ^* \just(n)$,}\\
                \text{\emph{(2)}} &
                \text{if $m \causes_\augP n^\varphi$, then $m \causes_\augP^* \just(n)$,}
        \end{array}
        \]
        where the events are annotated with their assumed polarity.
\end{lem}
\begin{proof}
        \emph{(1)} We must have $m = (1, m')$ and $n = (1, n')$ with $m' <_\augQ
        n'$, with $n'$ negative. By \emph{rule-abiding} and \emph{courtesy},
        $\just(n') \imc_\augQ n'$. Since $\augQ$ is a forest, $m' \leq_\augQ \just(n')$.

        \emph{(2)} Symmetric.
\end{proof}

Now, for any $m \in \ev{\augP \inter_\varphi \augQ}$, its \textbf{depth},
written $\depth(m)$, is $0$ if $m$ is minimal in $\deseq{\augP
\inter_\varphi \augQ}$, otherwise $\depth(\just(m)) + 1$. Now, we prove
the \emph{deadlock-free lemma}: 

\begin{lem}\label{lem:deadlock_free}
        Consider $\augQ \in \Aug(A\vdash B)$, $\augP \in \Aug(B \vdash
C)$ and $\varphi :
        x^\augQ_B \iso_B x^\augP_B$.

        Then $\causes$ is acyclic.
\end{lem}
\begin{proof}
        Seeking a contradiction, assume there is a cycle. By
        \autoref{lem:cycle_nf} it is entirely in $B$, without a
        minimal event in $B$. Writing it
        $\rho = m_1 \causes \dots \causes m_n
        \causes m_1$, define its \textbf{depth}
        \[
        \depth(\rho) = \sum_{i=1}^n \depth(m_i)\,,
        \]
        and \emph{w.l.o.g.} assume $\rho$ minimal for the product order on
        pairs $(n, d)$ where $d = \depth(\rho)$ and $n$ is its length. We
        notice that $\rho$ has no consecutive $\causes_\augQ$ or
        $\causes_\augP$ -- or we shorten the cycle by transitivity,
        breaking minimality. It also has no consecutive
        $\causes_\varphi$ by definition. 
        This entails that $n = 4k$, with \emph{w.l.o.g.}
        \[
        m_{4i}
        \causes_\augQ m_{4i+1}
        \causes_\varphi m_{4i+2}
        \causes_\augP m_{4i+3}
        \causes_\varphi m_{4i+4}\,,
        \]
        and from now on all indices are considered modulo $n$.

        Then $m_{4i+1}$ has polarity $\lside$. Otherwise, it has polarity
        $\varphi$, making $m_{4i+1} \causes_\varphi m_{4i+2}$
        impossible. Likewise, $m_{4i+3}$ has polarity $\rside$, while $m_{4i+2}$
        and $m_{4i+4}$ have polarity $\varphi$.

        We claim $\just(m_{4i+1}) \causes_\augQ \just(m_{4i})$. Observe
        $\just(m_{4i+1}) \imc_{\deseq{\augP \inter_\varphi \augQ}} m_{4i+1}$ 
        by definition; but by \emph{rule-abiding} this entails that
        $\just(m_{4i+1}) \causes_\augQ m_{4i+1}$. Since $\augQ$ is a forest,
        that makes $\just(m_{4i+1})$ comparable with $m_{4i}$ for
        $\causes_\augQ^*$. If $\just(m_{4i+1}) = m_{4i}$, then
        \[
        m_{4i-1} \causes_\augP m_{4i+2}
        \]
        since $\varphi$ is a symmetry and by
        \emph{rule-abiding} -- but this allows us to shorten the cycle,
        contradicting its minimality. Likewise, if $m_{4i} \causes_\augQ
        \just(m_{4i+1})$ then, observing that $\just(m_{4i+1})$ is of polarity
        $\varphi$ ($m_{4i+1}$ has polarity $\lside$), we obtain
        \[
        m_{4i} \causes_\augQ \just(\just(m_{4i+1}))
        \]
        by \autoref{lem:deadlock_free_aux} -- we cannot have an equality as
        they have distinct polarities. But then
        \[
        m_{4i} \causes_\augQ \just(\just(m_{4i+1})) \causes_\varphi
        \just(\just(m_{4i+2})) \causes_\augP m_{4i+3}
        \]
        yielding a cycle with the same length but strictly smaller depth,
        absurd. The last case remaining has $\just(m_{4i+1})
        \causes_\augQ m_{4i}$, but so $\just(m_{4i+1})
        \causes_\augQ \just(m_{4i})$ by \autoref{lem:deadlock_free_aux}
        (again, the equality is impossible for polarity reasons).

        With the same reasoning, $\just(m_{4i+3}) \causes_\augP
        \just(m_{4i+2})$; and $\just(m_{4i+2}) \causes_\varphi
        \just(m_{4i+1})$ and $\just(m_{4i+4}) \causes_\varphi
        \just(4i+3)$ by definition. So we can replace the whole cycle with
        \[
        \just(m_1) \causes \just(m_n) \causes \ldots \causes \just(m_1)
        \]
        reversing directions, with the same length but strictly smaller
        depth, contradiction.
\end{proof}

\section{Proof of the bialgebra distributivity law}
\label{app:proof_bialgebra}

Most laws for the bialgebra structure are straightforward; we
detail the exchange rule between $\delta_A$ and
$\mu_A$, which is more subtle.

For that, we must first introduce some notation: we write 
\[
\distribute_A \eqdef (\id_A \tensor \gamma_{A, A} \tensor \id_A) \odot
(\delta_A \tensor \delta_A)\,,
\qquad\qquad
\rassemble_A \eqdef \mu_A \tensor \mu_A
\]
for ``distribute'' and ``gather'',
which lets us phrase the desired bialgebra law as $\delta_A \odot \mu_A
= \rassemble_A \odot \distribute_A$. We skip the direct verification
that, writing $S(\posX_1,\ldots,\posX_n) = \nSymOf{\posX_1}
\times \ldots \times \nSymOf{\posX_n}$, 
\begin{eqnarray*}
\distribute_A &=& 
\sum_{\posX, \posY, \posU, \posV \in \pos(A)}
\frac{1}{S(\posX,\posY,\posU,\posV)}
\cdot
\bdistribute_A^{\posX, \posY, \posU, \posV}\\
\rassemble_A &=&
\sum_{\posX, \posY, \posU, \posV \in \pos(A)}
\frac{1}{S(\posX,\posY,\posU,\posV)}
\cdot
\brassemble_A^{\posX, \posY, \posU, \posV}
\end{eqnarray*}
where $\bdistribute_A^{\posX, \posY, \posU, \posV}$ and
$\brassemble_A^{\posX, \posY, \posU, \posV}$ are isomorphism
classes of $\distribute_A^{\rep{\posX}, \rep{\posY}, \rep{\posU},
\rep{\posV}}$ and $\rassemble_A^{\rep{\posX}, \rep{\posY}, \rep{\posU},
\rep{\posV}}$ for
\begin{eqnarray*}
\deseq{\distribute_A^{x,y,u,v}} &=& (x * y) \tensor (u * v) \vdash (x
\tensor u) \tensor (y \tensor v)\\
\deseq{\rassemble_A^{x,y,u,v}} &=& (x \tensor y) \tensor (u \tensor v)
\vdash (x * y) \tensor (u * v)
\end{eqnarray*}
with the obvious copycat behaviour on each component $x, y, u, v$.

We must introduce some additional notation. Consider $A$
a negative arena and $x \in \conf(A)$. We write $x = y \starplus z$
when $y, z \in \conf(A)$, $x = y \cup z$ and $y \cap z = \emptyset$ --
this is analogous to $x = y * z$ (and entails $x \sym y * z$), but
instead of the tagged disjoint union we have the standard set-theoretic
union, which happens to be disjoint. If $x = y \starplus z$
with $y \in \posY$ and $z \in \posZ$, we write
\[
x \darksplitinto \posY,  \posZ\,.
\]

But there may be several splittings of $x$ into $\posY$ and $\posZ$,
\emph{i.e.}\ pairs $(y, z)$ such that $x = y \starplus z$ with $y \in
\posY$ and $z \in \posZ$. We write $|x \darksplitinto \posY,  \posZ|$
the number of such pairs. It is easy to see that this is invariant
under symmetry, thus we may write $|\posX \darksplitinto \posY, \posZ$
for $|x \darksplitinto \posY,  \posZ|$ for any $x \in \posX$.

For our proof, the first key observation is the following lemma:

\begin{lem}\label{lem:main_distr}
Consider $A$ a negative arena, and $\posX,\posY,\posX',\posY' \in
\pos(A)$. Then,
\[
\bdelta_A^{\posX,\posY} \odot \bmu_A^{\posX',\posY'}
=
\quad\smashoperator{
\sum_{
\substack{
\posX_l, \posX_r, \posY_l, \posY_r
\text{~s.t.~}\\
\posX_l * \posX_r = \posX,~
\posY_l * \posY_r = \posY\\
\posX_l * \posY_l = \posX',~
\posX_r * \posY_r = \posY'
}
}}
\quad
|\posX \darksplitinto \posX_l, \posX_r|~
|\posY \darksplitinto \posY_l, \posY_r|~
|\posX' \darksplitinto \posX_l, \posY_l|~
|\posY' \darksplitinto \posX_r, \posY_r| \cdot 
\brassemble_A^{\posX_l, \posY_l, \posX_r, \posY_r}
\odot
\bdistribute_A^{\posX_l, \posX_r, \posY_l, \posY_r}\,.
\]
\end{lem}
\begin{proof}
Consider a symmetry $\varphi : x  * y \sym x' * y'$. This symmetry
sends some events of $x$ to $x'$, and some others to $y'$ -- likewise,
it sends some events of $y$ to $x'$, and some to $y'$. Following these
partitions, symmetries $\varphi : x  * y \sym x' * y'$ are in bijection
with 
\[
\left\{
\begin{array}{rcl}
x &=& x_l \starplus x_r\\
y &=& y_l \starplus y_r\\
x' &=& x'_l \starplus x'_r\\
y' &=& y'_l \starplus y'_r
\end{array}
\qquad,\qquad
\begin{array}{rcrcl}
\varphi_{l,l} &:& x_l &\sym& x'_l\\
\varphi_{l,r} &:& x_r &\sym& y'_l\\
\varphi_{r,l} &:& y_l &\sym& x'_r\\
\varphi_{r,r} &:& y_r &\sym& y'_r
\end{array}
\right\}\,,
\]

Additionally, this decomposition satisfies 
\[
\delta_A^{x,y} \odot_\varphi \mu_A^{x',y'}
=
\rassemble_A^{x'_l,x'_r,y'_l,y'_r} \odot_{(\varphi_{l,l} \tensor
\varphi_{r,l}) \tensor (\varphi_{l,r} \tensor \varphi_{r,r})}
\distribute_A^{x_l,x_r,y_l,y_r}
\]
which is verified by an immediate analysis of the copycat behaviour of
this composition. 

To prove the lemma, we then proceed with, for arbitrary $x \in \posX, y \in
\posY, x' \in \posX', y' \in \posY'$:
\begin{align*}
&\bdelta_A^{\posX,\posY} \odot \bmu_A^{\posX',\posY'}
\displaybreak[0]\\
&=
\sum_{\varphi : x * y \sym x' * y'}
\ic{\delta_A^{x,y} \odot_\varphi \mu_A^{x',y'}}
\displaybreak[1]\\
&= \sum_{
\substack{x = x_l \smallstarplus x_r\\
y = y_l \smallstarplus y_r\\
x' = x'_l \smallstarplus x'_r\\
y' = y'_l \smallstarplus y'_r
}
}
\sum_{
\substack{
\varphi_{l,l} : x_l \sym x'_l\\
\varphi_{l,r} : x_r \sym y'_l\\
\varphi_{r,l} : y_l \sym x'_r\\
\varphi_{r,r} : y_r \sym y'_r
}
}
\ic{
\rassemble_A^{x'_l,x'_r,y'_l,y'_r} \odot_{(\varphi_{l,l} \tensor
\varphi_{r,l}) \tensor (\varphi_{l,r} \tensor \varphi_{r,r})}
\distribute_A^{x_l,x_r,y_l,y_r}}
\displaybreak[2]\\
&= 
\qquad\smashoperator[l]{
\sum_{
\substack{
\posX_l, \posX_r, \posY_l, \posY_r
\text{~s.t.~}\\
\posX_l * \posX_r = \posX,~
\posY_l * \posY_r = \posY\\
\posX_l * \posY_l = \posX',~
\posX_r * \posY_r = \posY'
}
}}
\sum_{
\substack{
x = x_l \smallstarplus x_r\text{~s.t.~}x_l \in \posX_l, x_r \in \posX_r\\
y = y_l \smallstarplus y_r\text{~s.t.~}y_l \in \posY_l, y_r \in \posY_r\\
x' = x'_l \smallstarplus x'_r\text{~s.t.~}x'_l \in \posX'_l, x'_r \in \posX'_r\\
y' = y'_l \smallstarplus y'_r\text{~s.t.~}y'_l \in \posY'_l, y'_r \in \posY'_r
}
}
\sum_{
\substack{
\varphi_{l,l} : x_l \sym x'_l\\
\varphi_{l,r} : x_r \sym y'_l\\
\varphi_{r,l} : y_l \sym x'_r\\
\varphi_{r,r} : y_r \sym y'_r
}
}
\ic{
\rassemble_A^{x'_l,x'_r,y'_l,y'_r} \odot_{(\varphi_{l,l} \tensor
\varphi_{r,l}) \tensor (\varphi_{l,r} \tensor \varphi_{r,r})}
\distribute_A^{x_l,x_r,y_l,y_r}}
\displaybreak[2]\\
&=
\sum_{
\substack{
\posX_l, \posX_r, \posY_l, \posY_r
\text{~s.t.~}\\
\posX_l * \posX_r = \posX,~
\posY_l * \posY_r = \posY\\
\posX_l * \posY_l = \posX',~
\posX_r * \posY_r = \posY'
}
}
\sum_{
\substack{
x = x_l \smallstarplus x_r\text{~s.t.~}x_l \in \posX_l, x_r \in \posX_r\\
y = y_l \smallstarplus y_r\text{~s.t.~}y_l \in \posY_l, y_r \in \posY_r\\
x' = x'_l \smallstarplus x'_r\text{~s.t.~}x'_l \in \posX'_l, x'_r \in \posX'_r\\
y' = y'_l \smallstarplus y'_r\text{~s.t.~}y'_l \in \posY'_l, y'_r \in \posY'_r
}
}
\brassemble_A^{\posX_l, \posY_l, \posX_r, \posY_r}
\odot
\bdistribute_A^{\posX_l, \posX_r, \posY_l, \posY_r}\\
\displaybreak[0]\\
&= 
\sum_{
\substack{
\posX_l, \posX_r, \posY_l, \posY_r
\text{~s.t.~}\\
\posX_l * \posX_r = \posX,~
\posY_l * \posY_r = \posY\\
\posX_l * \posY_l = \posX',~
\posX_r * \posY_r = \posY'
}
}
|\posX \darksplitinto \posX_l, \posX_r|\ 
|\posY \darksplitinto \posY_l, \posY_r|\ 
|\posX' \darksplitinto \posX_l, \posY_l|\ 
|\posY' \darksplitinto \posX_r, \posY_r| \cdot
\brassemble_A^{\posX_l, \posY_l, \posX_r, \posY_r}
\odot
\bdistribute_A^{\posX_l, \posX_r, \posY_l, \posY_r}\,.
\end{align*}
by the definition of composition of isogmentations (which 
does not depend on the chosen representative); the observation above;
reorganizing the sum by symmetry classes; again via the definition of
composition of isogmentations; and by definition of
splittings.
\end{proof}

This is sufficient to ensure that $\delta_A \odot \mu_A$ and
$\rassemble_A \odot \distribute_A$ have the same isogmentations, but
not that they occur with the same coefficient. For that, the next key
observation is:

\begin{lem}\label{lem:other_distr}
Consider $A$ a negative arena, and $\posX, \posY \in \pos(A)$. Then,
\[
\nSymOf{\posX * \posY} = | \posX * \posY \darksplitinto \posX, \posY |
\times
\nSymOf{\posX} \times \nSymOf{\posY}\,.
\]
\end{lem}
\begin{proof}
Fix arbitrary $x \in \posX$ and $y \in \posY$ that we assume disjoint,
and $z = x \starplus y \in \posX * \posY$. The set of symmetries on
$\posX * \posY$ is clearly in bijection with the set of symmetries
\[
\varphi : z \sym x * y
\]
which we shall study. As in the lemma above, such a symmetry sends some
events of $z$ to $x$ and some to $y$; this induces a splitting $z = x'
\starplus y'$ with induced $\varphi_x : x' \sym x$ and $\varphi_y : y'
\sym y$, so that $x' \in \posX$ and $y' \in \posY$. Conversely, any
such splitting of $z$ with accompanying symmetries yields a symmetry $z
\sym x * y$. From this it is straightforward to obtain a bijection
witnessing the announced equality (keeping in mind that we may fix in
advance a chosen $\kappa_x : x \sym \rep{\posX}$ for all $x \in \posX$,
so as to bridge between symmetries $x' \sym x$ and endosymmetries
$\rep{\posX} \sym \rep{\posX}$).
\end{proof}

With this in place, we may finally prove the exchange law of
bialgebras:

\begin{align*}
& \delta_A \odot \mu_A
\\
&=
\sum_{
\substack{
\posX, \posY \in \pos(A)\\
\posX', \posY' \in \pos(A)}}
\frac{1}{S(\posX, \posY, \posX', \posY')} 
\cdot
\bdelta_A^{\posX,\posY} \odot \bmu_A^{\posX',\posY'}
\displaybreak[0]\\
&=
\sum_{
\substack{
\posX, \posY \in \pos(A)\\
\posX', \posY' \in \pos(A)}}
\sum_{
\crampedsubstack{
\posX_l, \posX_r, \posY_l, \posY_r
\text{~s.t.~}\\
\posX_l * \posX_r = \posX,~
\posY_l * \posY_r = \posY\\
\posX_l * \posY_l = \posX',~
\posX_r * \posY_r = \posY'
}
}
\kern -13pt\frac{
    |\posX \darksplitinto \posX_l, \posX_r|
    \ |\posY \darksplitinto \posY_l, \posY_r|
    \ |\posX' \darksplitinto \posX_l, \posY_l|
    \ |\posY' \darksplitinto \posX_r, \posY_r|
}{S(\posX, \posY, \posX', \posY')}
\cdot 
\brassemble_A^{\posX_l, \posY_l, \posX_r, \posY_r}
\odot
\bdistribute_A^{\posX_l, \posX_r, \posY_l, \posY_r}
\displaybreak[2]\\
&=
\sum_{
\posX_l, \posX_r, \posY_l, \posY_r \in \pos(A)
}
\frac{
    |\posX \darksplitinto \posX_l, \posX_r|
    \ |\posY \darksplitinto \posY_l, \posY_r|
    \ |\posX' \darksplitinto \posX_l, \posY_l|
    \ |\posY' \darksplitinto \posX_r, \posY_r|
}{S(\posX_l * \posX_r, \posY_l
* \posY_r, \posX_l * \posY_l, \posX_r * \posY_r)}
\cdot 
\brassemble_A^{\posX_l, \posY_l, \posX_r, \posY_r}
\odot
\bdistribute_A^{\posX_l, \posX_r, \posY_l, \posY_r}
\displaybreak[2]\\
&=
\sum_{\posX_l,\posX_r,\posY_l,\posY_r \in \pos(A)}
\frac{1}
{S(\posX_l,\posX_r,\posY_l,\posY_r)^2}
\cdot
\brassemble_A^{\posX_l, \posY_l, \posX_r, \posY_r}
\odot
\bdistribute_A^{\posX_l, \posX_r, \posY_l, \posY_r}
\displaybreak[2]\\
&=
\left(
\sum_{\posX_l,\posX_r,\posY_l,\posY_r \in \pos(A)}
\frac{1}
{S(\posX_l,\posX_r,\posY_l,\posY_r)}
\cdot \brassemble_A^{\posX_l, \posY_l, \posX_r, \posY_r}
\right)
\odot 
\left(
\sum_{\posX_l,\posX_r,\posY_l,\posY_r \in \pos(A)}
\frac{1}
{S(\posX_l,\posX_r,\posY_l,\posY_r)}
\cdot \bdistribute_A^{\posX_l, \posY_l, \posX_r, \posY_r}
\right)\\
&= \rassemble_A \odot \distribute_A
\end{align*}
by first unfolding the definitions of $\delta_A$ and $\mu_A$; then by
\autoref{lem:main_distr}; then re-indexing the sum in the obvious way;
then applying \autoref{lem:other_distr}; and finally by
linearity, observing that all the compositions
$\brassemble_A^{\posX_l, \posY_l, \posX_r, \posY_r}
\odot
\bdistribute_A^{\posX'_l, \posX'_r, \posY'_l, \posY'_r}$ where one of
$\posX_l = \posX'_l, \posY_l = \posX'_r, \posX_r = \posY'_l$ and
$\posY_r = \posY'_r$ does not hold is null.

\bibliographystyle{alphaurl}
\bibliography{main}

\end{document}